\pdfoutput=1
\documentclass[a4paper,reprint,twocolumn,notitlepage,groupedaddress,aip,nofootinbib]{revtex4-2}

\usepackage[left=1.5cm,right=1.5cm,top=1cm,bottom=1.5cm,includeheadfoot]{geometry}

\usepackage[english]{babel}
\usepackage[utf8x]{inputenc}
\usepackage{tgtermes} 
\usepackage{tgheros} 
\usepackage[T1]{fontenc}

\usepackage{amssymb}
\usepackage{amsmath}
\usepackage{amsthm}
\usepackage{mathtools}
\usepackage[dvipsnames]{xcolor}
\usepackage[shortlabels]{enumitem}
\usepackage[normalem]{ulem}
\PassOptionsToPackage{hyphens}{url} 
\usepackage[breaklinks=true]{hyperref}
\bibpunct{[}{]}{;}{n}{}{} 

\usepackage{graphicx}
\usepackage[newcommands]{ragged2e} 
\usepackage{subcaption}

\usepackage{tikz}
\usetikzlibrary{shapes.geometric}
\usetikzlibrary{calc}
\usetikzlibrary{matrix}
\usetikzlibrary{patterns}
\usetikzlibrary{arrows.meta}
\tikzset{
  graphnode/.style={draw,circle,fill=SkyBlue,draw=black}
}

\theoremstyle{plain}
\newtheorem{theorem}{Theorem}

\newtheorem{lemma}[theorem]{Lemma}
\newtheorem{corollary}[theorem]{Corollary}

\theoremstyle{definition}

\theoremstyle{remark}

\newcommand{\R}{\mathbb{R}}
\newcommand{\C}{\mathbb{C}}

\renewcommand{\i}{\mathrm{i}}

\renewcommand{\H}{\mathcal{H}}
\newcommand{\U}{\mathcal{U}}

\newcommand{\Dens}{\mathsf{Dens}}
\newcommand{\Pot}{\mathsf{Pot}}
\newcommand\orho[1]{\rho_{\overline{#1}}}
\newcommand\oirho[2]{\rho_{\overline{#1},#2}}
\newcommand\orhoRe[1]{\rho_{\overline{#1}}^\mathsf{Re}}
\newcommand\orhoIm[1]{\rho_{\overline{#1}}^\mathsf{Im}}
\newcommand{\vA}{v_\mathrm{A}}

\DeclareMathOperator{\Tr}{Tr}

\DeclareMathOperator{\ch}{ch}
\DeclareMathOperator{\seg}{seg}
\DeclareMathOperator{\rank}{rank}
\DeclareMathOperator{\linearspan}{span}

\renewcommand{\Re}{\mathop\mathsf{Re}}
\renewcommand{\Im}{\mathop\mathsf{Im}}
\newcommand\T[1]{{#1}^\mathsf{T}}

\begin{document}

\title{
    Geometry of Degeneracy in Potential and Density Space
}

\author{Markus Penz}
\address{Basic Research Community for Physics, Innsbruck, Austria}
\email{m.penz@inter.at}

\author{Robert van Leeuwen}
\address{Department of Physics, Nanoscience Center, University of Jyv\"askyl\"a, Finland}

\begin{abstract}
    In a previous work [\href{https://doi.org/10.1063/5.0074249}{J.\ Chem.\ Phys.\ \textbf{155}, 244111 (2021)}], we found counterexamples to the fundamental Hohenberg--Kohn theorem from density-functional theory in finite-lattice systems represented by graphs. Here, we demonstrate that this only occurs at very peculiar and rare densities, those where density sets arising from degenerate ground states, called degeneracy regions, touch each other or the boundary of the whole density domain. Degeneracy regions are shown to generally be in the shape of the convex hull of an algebraic variety, even in the continuum setting.
    The geometry arising between density regions and the potentials that create them is analyzed and explained with examples that, among other shapes, feature the Roman surface.
\end{abstract}
\maketitle

\section{Introduction}
\label{sec:intro}

The mapping from potentials to the corresponding one-particle ground-state density in a many-particle quantum system and the possibility of its inversion forms the basis of density-functional theory (DFT) \cite{vonBarth2004basic,burke2007abc,dreizler2012-book,eschrig2003-book} and many of its variants \cite{ullrich2011-book,ullrich2014brief,Vignale1987,VIGNALE_PRB70_201102,ruggenthaler2014-QEDFT}. This theory is widely employed in electronic-structure calculations, allowing for a good balance between accuracy and computational cost. The problem can also be viewed as a control problem, where a potential that produces a given target density is sought.
Despite its prominent role, the topological structure of the density-potential mapping is poorly studied, even within ground-state DFT. Notable exceptions are a work by \citet{ullrich2002} and a recent paper by \citet{garrigue2021-potential-density}. The former investigated the dimensionality of sets in potential space preserving a certain degeneracy in a finite-lattice system. This indicated that while degeneracy is quite common in the density domain, it remains rare in potential space (yet, the latter statement does not follow in their proof that fails to consider possible linear dependencies of conditions, see Section~\ref{sec:Ullrich-Kohn}).
Such finite-lattice systems are a frequent object of investigation in quantum mechanics, especially in solid-state physics, where their most basic realization is the important Hubbard model \cite{arovas2022hubbard,qin2022hubbard}. In DFT, finite-lattice systems also naturally arise with the choice of a finite basis of localized orbitals and consequently are of prime importance for the whole field \cite{flores2022localorbital}.

The authors of the paper at hand also started their inquiries in this area by focusing on finite-lattice systems, generalized by graphs, and a number of surprising results were already found, including the nullity of the celebrated Hohenberg--Kohn theorem for systems with special symmetry \cite{DFT-graphs}.
The studied examples suggested a far-reaching conjecture on the special geometry of the mapping, about the connection between densities that arise from degenerate ground states, later called \emph{degeneracy regions}, and those for which the Hohenberg--Kohn theorem fails, as well as their corresponding potentials. This conjecture will be proven here (Section~\ref{sec:geom}), but first the shape of such degeneracy regions will be clarified, with the surprising discovery of highly intricate objects from algebraic geometry, e.g., the Roman surface, for which we give a basic classification (Sections~\ref{sec:subspaces-density-map}-\ref{sec:density-regions}). With the same techniques we give precise conditions where non-pure-state $v$-representable densities, a concept from DFT, arise (Section~\ref{sec:non-pure-v-rep}). To establish a link to the full geometry of the mapping, the study of \citet{ullrich2002} is rendered more precise and a new proof based on Rellich's theorem is provided (Section~\ref{sec:Ullrich-Kohn}). The examples are all for small lattice systems and spinless fermions, but it must be stressed that the results from Sections~\ref{sec:subspaces-density-map}-\ref{sec:non-pure-v-rep} also apply to continuum systems and everything can also be extended to incorporate spin by simply including additional internal degrees of freedom. Since this work can be considered a sequel to our previous paper \cite{DFT-graphs}, we recommend looking there for a deepened understanding.

The vector space that includes one-particle densities shall generally be denoted $\Dens$ and similarly $\Pot$ is the vector space of one-body potentials. In a lattice system with $M$ vertices this means that $\Dens = \Pot = \R^M$, while in a continuum setting one can choose $\Dens = L^1(\R^3)\cap L^3(\R^3)$ and $\Pot = L^\infty(\R^3) + L^{3/2}(\R^3)$, its topological dual, like in \citet{Lieb1983}.
We define the density map $\rho:\H \to \Dens, \Psi \mapsto \langle\Psi,\hat\rho\Psi\rangle / \|\Psi\|^2$ (also we use  $\rho:\mathcal{D}\to\Dens, \Gamma\mapsto\Tr(\hat\rho\Gamma)$  on the set of all density matrices for ensemble states $\mathcal{D}$)  that takes fermionic many-particle states in Hilbert space to their respective one-particle density via the density operator $\hat\rho$ \cite{eschrig2003-book}. We study a class of Hamiltonians $H_v = H_0 + V$ that only differ with respect to their real, scalar one-body potential $v$. Here, $H_0$ is the fixed (internal) part that is always assumed to be real (i.e., it does not include a vector potential, or, more generally, allows for time-reversal symmetry), while $V$ is the operator acting on $\H$ that corresponds to the one-body potential $v$. The following lemma will be important for the choice of eigenvectors of such Hamiltonians.

\begin{lemma}\label{lem:real-eigenstate}
If $H_v$ is real symmetric and has a $g$-dimensional eigenspace $\mathcal{U}$ with eigenvalue $E$ then this space is spanned by $g$ \emph{real} orthonormal vectors $\{\Phi_k\}_{k=1}^g$ with \emph{complex} coefficients, $\mathcal{U} = \linearspan_\C \{ \Phi_1,\ldots,\Phi_g \}$.
\end{lemma}

\begin{proof}
For the $g$-dimensional eigenspace $\mathcal{U}$ with eigenvalue $E$ choose a general basis $\{\Psi_k\}_{k=1}^g$. Then observe that since $H_v$ and $E$ are real, the real and imaginary part of those vectors are eigenvectors as well, $H_v\Re\Psi_k=E\Re\Psi_k$ and $H_v\Im\Psi_k=E\Im\Psi_k$. This gives a total of $2g$ real eigenvectors that span the whole $\mathcal{U}$ with complex coefficients, from which one can get $g$ real and orthonormal eigenvectors $\{\Phi_k\}_{k=1}^g$ by the Gram--Schmidt process that span the same space.
\end{proof}

In particular, we can always already choose an orthonormal basis of real eigenstates to span the ground-state eigenspace, a property that we will use extensively below.
The following theorem allows us in principle to switch any discussion about ground states over to just densities. This implies an enormous reduction of complexity.

\begin{theorem}\label{th:HK1}
Assume that two Hamiltonians $H_v = H_0 + V,H_{v'} = H_0 + V'$ that differ only in their scalar one-body potentials $v,v'$ share a common ground-state density $\rho$. Then an (ensemble) ground state $\Psi$ $(\Gamma)$ of $H_v$ with density $\rho$ is also an (ensemble) ground state of $H_{v'}$ and vice versa.
\end{theorem}

This result means that to every ground-state density $\rho$ we can assign a ground state $\Psi(\rho)$ (or, more generally, a class of ground states or ensembles with the same density), \emph{irrespective} of the present potential. As such, this is well-known as a part of the Hohenberg--Kohn theorem and is individually called a \emph{weak HK-like result} sometimes \cite{Tellgren2018}. We refer to our previous work \cite{DFT-graphs} and a recent review article \cite{dens-pot-review} for two distinct proofs of the above theorem. Already Theorem~\ref{th:HK1} allows to define a universal energy functional $F_\mathrm{HK}(\rho) = \langle \Psi(\rho),H_0\Psi(\rho) \rangle$ for every ground-state density $\rho$ that gives the lowest possible internal (kinetic plus interactions) energy of a pure state with the prescribed density. It is customary to instead use the much more well-behaved constrained-search functional over ensemble states $F(\rho) = \inf_{\Gamma\mapsto\rho} \Tr(H_0\Gamma)$ \cite{Lieb1983} that is convex and equal to $F_\mathrm{HK}$ on the latter functional's domain. Any ground-state density $\rho$ to a potential $v$ will then minimize the functional $F(\rho)+\langle v,\rho\rangle$ since this gives the total energy $E(v)$. An excellent mathematical review on the universal functionals in DFT is \citet{LewinFunctionals}.

The second part of the Hohenberg--Kohn theorem, which then maps ground states to unique potentials (modulo a constant) and in conjunction establishes the mapping from ground-state densities to potentials, is omitted here. While it holds in the continuum case for a large class of potentials \cite{Garrigue2018}, it generally fails for finite lattice systems \cite{DFT-graphs} that will serve as the prime examples in this work. Where it fails, we find densities that can be represented by multiple potentials that differ in more than an additive constant. Such densities we call \emph{non-uniquely $v$-representable} (``\emph{non-uv}'') and one observes that they always arise when degeneracy regions touch, a feature that is the main element of our geometry theorem in Section~\ref{sec:geom}. Thus we start by investigating the shape of such density sets.

\section{State subspaces under the density map}
\label{sec:subspaces-density-map}

Let $\mathcal{U}$ be a $g$-dimensional subspace of a given Hilbert space $\H$ spanned by \emph{real} vectors $\{\Phi_k\}_{k=1}^g$ with \emph{complex} coefficients. The Hilbert space can be $\H = \C^L$, $L=\binom{ M }{ N }$, in the case of $N$ spinless fermionic particles on an $M$-vertex lattice \cite[Sec.~II.B]{DFT-graphs}, or the usual anti-symmetric many-particle Hilbert space for continuum systems $\H = \Lambda^N L^2(\R^3)$. It is possible to extend to particles with spin by including the internal degrees of freedom into the Hilbert space of the individual particles.
Later, the subspace $\mathcal{U}$ will be the space of $g$-fold degenerate ground states of a Hamiltonian $H_v$ and so we call $g$ the \emph{degree}. We are interested in the image of $\U$ under $\rho$.
Let $x=(x_1,\ldots,x_g) \in\C^g$, $\|x\|=1$, be the coordinates for a normalized state $\Psi\in\mathcal{U}$ with respect to the real, orthonormal subspace basis $\{\Phi_k\}_{k=1}^g$ of $\mathcal{U}$. The density is then evaluated as
\begin{align}\label{eq:rho-from-x}
    \rho(\Psi) = \sum_{k=1}^g |x_k|^2 \rho(\Phi_k) + \sum_{\substack{k,l=1\\k< l}}^g 2\Re ( x_k^*x_l ) \langle \Phi_k,\hat\rho\Phi_l\rangle
\end{align}
and we will later also use the notation $\rho(x)$ for a fixed basis $\{\Phi_k\}_{k=1}^g$. We can also define the density map for ensembles $\Gamma\in\mathcal{D}$ in $\mathcal{U}$, i.e., $\Gamma\H \subseteq \mathcal{U}$, that we call $\rho$ as well. If $P_j$ are the projections on orthonormal pure states $\Psi_j$ from the subspace $\mathcal{U}$ and $\Gamma=\sum_j w_j P_j$ is a density matrix with coefficients $w_j\geq 0$, $\sum_j w_j =1$, then we define accordingly
\begin{equation}\label{eq:rho-from-Gamma}
    \rho (\Gamma)=\sum_j w_j \rho (\Psi_j).
\end{equation}

For the further discussion we define the concept of a \emph{density region} as the set of all densities belonging to states in $\U$. Later this will be a \emph{degeneracy region}, the set of all densities belonging to states in the ground-state eigenspace $\U$ for a given potential $v$. Therein, we differentiate three levels, by first mapping only $\mathcal{U}_\R = \linearspan_\R \{ \Phi_1,\ldots,\Phi_g \}$, i.e., the subspace spanned by the basis vectors with only real expansion coefficients, then the linear span $\mathcal{U} = \linearspan_\C \{ \Phi_1,\ldots,\Phi_g \}$ for general, complex expansion coefficients (which just means the whole subspace), and lastly by forming the full density region as all densities from that subspace including mixed states.
\begin{subequations}
\begin{align}
    &D_\R = \rho(\mathcal{U}_\R) \label{eq:D_R}\\
    &D_\C = \rho(\U) \label{eq:D_C}\\
    &D = \rho(\{ \Gamma \in \mathcal{D} \mid \Gamma\H \subseteq \U \}) \label{eq:D}
\end{align}
\end{subequations}
From the definition we see that the density region does not depend on the choice of basis for $\mathcal{U}$ and that $D_\C$ is limited to pure-state densities, while $D$ also includes non-pure-state densities (see Section~\ref{sec:non-pure-v-rep} where it is shown that in general $D_\C \neq D$).
That $D_\R \subseteq D_\C$ follows by definition and $D$ includes the previous two sets, since $D$ is the set of all convex combinations of pure-state densities as expressed in \eqref{eq:rho-from-Gamma}. We thus find the following sequence of inclusions,
\begin{equation} \label{eq:D-inclusions}
    D_\R \subseteq D_\C \subseteq D.
\end{equation}
As the continuous image of a compact and connected set those sets are all compact and connected.
In order to get $D_\C$ from $D_\R$ we define the \emph{segment set} for an arbitrary set $X$ in a vector space as
\begin{equation}
    \seg X = \{ \lambda x+(1-\lambda) y \mid x,y\in X, 0\leq \lambda \leq 1 \}.
\end{equation}
Hence, instead of taking an arbitrary, finite number of points and form their convex combination, like in the convex hull, we only take two points in the segment set. The segment set is equivalent to the geometric join, as introduced in \citet{BK_2012}, of a set with itself. The construction of $D_\C$ from $D_\R$ is then furnished by the following lemma.

\begin{lemma}
\label{lem:DC=segDR}
    $D_\C = \seg D_\R$.
\end{lemma}

\begin{proof}
Take any density in $D_\C$, then it must be the density of a normalized $\Psi \in \mathcal{U}$. Such a vector can always be split up into its real and imaginary part as $\Psi=\sqrt{\lambda}\Psi_1 + \i\sqrt{1-\lambda}\Psi_2$ with $0\leq\lambda\leq 1$ and $\Psi_1,\Psi_2 \in \linearspan_\R \{ \Phi_1,\ldots,\Phi_g \}$, both normalized as well. Now,
\begin{equation}
\begin{aligned}
    \rho(\Psi) &= \langle \Psi,\hat\rho\Psi \rangle = \lambda\langle \Psi_1,\hat\rho\Psi_1 \rangle + (1-\lambda) \langle \Psi_2,\hat\rho\Psi_2 \rangle \\
    &= \lambda \rho(\Psi_1) + (1-\lambda) \rho(\Psi_2),
\end{aligned}
\end{equation}
where the mixed term cancels since $\Psi_1$ and $\Psi_2$ are real. Since $\rho(\Psi_1),\rho(\Psi_2) \in D_\R$ this proves the assertion.
\end{proof}

By \eqref{eq:rho-from-Gamma} the $\rho(\Gamma)$ is a convex combination of pure-state densities and it holds $D=\ch D_\C$, the convex hull, so the following result follows directly from Lemma~\ref{lem:DC=segDR}.

\begin{corollary}\label{cor:D=chDR}
$D=\ch D_\R$.
\end{corollary}

This greatly simplifies the analysis of a density region $D$, since we can limit ourselves to the study of $D_\R$ and then just form the convex hull.
In order to get $D_\R$, we choose $x\in\R^g$, $\|x\|=1$, and split the map $x\mapsto\rho(x)$ from coordinates to densities in \eqref{eq:rho-from-x} into two stages. First, form the $(g+1)g/2$ second-order monomials $x_k^2$ and $x_kx_l$ that then live in a larger space,
\begin{equation}\begin{aligned}
    \nu : \R^g \quad\quad\quad\;\; &\longrightarrow \R^{(g+1)g/2} \\
    (x_1,\ldots,x_g) &\longmapsto (x_1^2,\ldots,x_g^2,x_1x_2,\ldots,x_{g-1}x_g).
\end{aligned}\end{equation}
This map is known as the Veronese embedding in algebraic geometry \cite{beltrametti-book,Harris_book} and its image of the unit sphere is the Veronese variety $\mathbb{V}_g$.
We immediately see that antipodal points $\pm x$ on the unit sphere are mapped to the same vector, so the map is usually studied within projective space. Great circles on the unit sphere are mapped to a ellipses \cite[Prop.~3.3]{degen1994}.
The second stage consists of \emph{linearly} mapping to densities with a map $P$ formed by the factors $\orho{k}=\rho(\Phi_k)$ and $\orho{kl} = 2\langle \Phi_k,\hat\rho\Phi_l\rangle$ 
as given in \eqref{eq:rho-from-x},
\begin{equation}\label{eq:P-nu-map}\begin{aligned}
    \rho=\,&P \circ \nu : \R^g \;\;\, \longrightarrow \Dens \\
    &(x_1,\ldots,x_g) \longmapsto \sum_{k=1}^g x_k^2 \orho{k} + \sum_{\substack{k,l=1\\k<l}}^g x_k x_l \orho{kl}.
\end{aligned}\end{equation}
The notation with the overlined index has been introduced to avoid confusion with the density $\rho_i$ at a certain lattice point.
This already shows that $D_\R$ is the linear map $P$ applied to $\mathbb{V}_g$ and thus forms a parametrized algebraic variety, while the whole density region $D$ is then just its convex hull. In order to arrive at a possible classification for $D$, we determine its dimension within density space. This is given by the number of linearly independent factors $\orho{k}$ and $\orho{kl}$, that is a total of $(g+1)g/2$ minus the dimension of the kernel (nullity) of $P$ given by $\kappa=\dim\ker P$, and then finally minus 1 from the normalization constraint for densities,
\begin{equation}\label{eq:max-dim-D}
    \dim D = \frac{1}{2}(g+1)g -\kappa -1.
\end{equation}
Hence, the basic classification of density regions will be in $(g,\kappa)$, the \emph{degree} $g \in \{2,3,\ldots\}$ and the \emph{nullity} $\kappa \in \{0,\ldots,(g+1)g/2-1\}$. That $\kappa = (g+1)g/2-1$ is possible in principle, but this implies that every state in $\U$ is mapped to the exactly same density, which seems unlikely in the absence of internal degrees of freedom. This also implies that all $\orho{kl}$ must be zero, since they sum up to zero and cannot be equal to the density. The same two assertions also appear in \citet[Cor.~1.7]{garrigue2021-potential-density} for the conservation of degeneracy in certain directions of potential variation. The density region is then just a single point, a \emph{singleton set}, and no special geometry arises. A larger $\kappa$ cannot occur because then $\rho$ would be the zero map. We can now imagine a mapping $\deg : \rho \mapsto g$ that assigns to every element of $\Dens$ that is the ground-state density for some potential $v$ (``$v$-representable'') the degree of the degeneracy region it belongs to. The ``rich structure'' \cite{garrigue2021KS} of this map will become apparent in the following examples.

\section{Further classification and examples for density regions}
\label{sec:density-regions}

In this section the setting is the same as before, with $\mathcal{U}$ a $g$-dimensional subspace of a Hilbert space $\H$ that corresponds to a lattice or a continuum system. The aim is to study and classify the density regions originating from $\mathcal{U}$.
We begin with the lowest degree $g=2$, so the possible nontrivial values for $\kappa$ are either 0 or 1. Since $x_1^2+x_2^2=1$ on the unit circle, we choose $x_1=\cos\varphi, x_2=\sin\varphi$ and \eqref{eq:P-nu-map} transforms into
\begin{equation}
\begin{aligned}
    &\rho(x) = (\cos\varphi)^2 \orho{1} + (\sin\varphi)^2 \orho{2} + \sin\varphi\,\cos\varphi\,\orho{12} \\
    &= \frac{1+\cos(2\varphi)}{2}\orho{1} + \frac{1-\cos(2\varphi)}{2}\orho{2} + \frac{\sin(2\varphi)}{2}\orho{12} \\
    &= \frac{\orho{1}+\orho{2}}{2} + \frac{\orho{1}-\orho{2}}{2} \cos(2\varphi) + \frac{\orho{12}}{2}\sin(2\varphi).
\end{aligned}
\end{equation}
Here, the doubling of the angle $\varphi$ is a result of the identification of antipodal points.
The formula shows that all densities in $D_\R$ lie on an ellipse with center $\overline{\rho}=\frac{1}{2}(\orho{1}+\orho{2})$ and axes $\frac{1}{2}(\orho{1}-\orho{2})$ and $\frac{1}{2}\orho{12}$. Consequently, $D$ will be the filled ellipse. If $\kappa=1$ then all three factors $\orho{1},\orho{2},\orho{12}$ span the image of $P$ as a 2-dimensional plane in density space, so $\overline{\rho}$, $\frac{1}{2}(\orho{1}-\orho{2})$, $\frac{1}{2}\orho{12}$ and with them the whole elliptical $D$ lie in a plane through the origin. But the normalization constraint for densities subtracts one further dimension, so $D$ actually collapses into a line segment, as required by \eqref{eq:max-dim-D}. If $\kappa=0$ this restriction does not occur and $D$ will really be a filled ellipse. A simple example for this case is the ground-state eigenspace for two non-interacting spinless particles on a lattice with three vertices in the shape of a triangle that is extensively discussed in Sec.~VI.C of \citet{DFT-graphs}. There $D$ forms the incircle of the triangular density domain that itself is the convex hull of the three extremal densities $(1,1,0)$, $(1,0,1)$ and $(0,1,1)$.

If $g=3$ and $\kappa \geq 2$, the maximal dimension of the density region given by \eqref{eq:max-dim-D} is 3, so we still have a possibility to easily visualize it. The set $D_\R$ as the image of $\mathbb{V}_3$ under $P$ with $\kappa=2$ is generally called a \emph{Steiner surface} \cite{apery-book}.
Here we will limit the discussion to one example which yields the most famous such surfaces, the \emph{Roman surface}, while it remains open if other types even appear as density regions. For this example we utilize again a finite-lattice system with two non-interacting spinless particles, now in the form of a tetrahedron graph with equal hopping between all vertex pairs. Since the density is defined on just four vertices, the density space is $\R^4$ while the Pauli exclusion principle and the normalization of the density additionally yield the restrictions $0\leq \rho_i\leq 1$ and $\sum_i\rho_i=2$. This gives a density domain in the shape of an octahedron within the three-dimensional affine hyperplane defined by the normalization constraint. The density region $D_\R$ to be determined here will then be a set within this octahedron. The one-particle Hamiltonian for the given system when derived from a graph Laplacian \cite[Sec.~II.D]{DFT-graphs} is
\begin{equation}
h_0 =
\begin{pmatrix*}[r]
3 & -1 & -1 & -1 \\
-1 & 3 & -1 & -1 \\
-1 & -1 & 3 & -1 \\
- 1 & -1 & -1 & 3 
\end{pmatrix*}
\end{equation}
and for $v=0$ has the ground-state orbital
\begin{equation}
\phi_0 = \tfrac{1}{2} (1,1,1,1)
\end{equation}
as well as the 3-fold degenerate first-excited states
\begin{align}
\phi_1 &= \tfrac{1}{2} (-1,-1,1,1) \\
\phi_2 &= \tfrac{1}{2} (-1,1,-1,1) \\
\phi_3 &= \tfrac{1}{2} (1,-1,-1,1) 
\end{align}
that all give rise to a uniform density. These orbitals yield the two-particle Slater determinants $\Phi_k = \phi_0 \wedge \phi_k$, $k\in \{1,2,3\}$, that span the three-dimensional ground-state subspace $\mathcal{U}$.
We note here that for Hamiltonians derived from such graph Laplacians it generally follows from the Perron--Frobenius theorem that the ground-state orbital is non-degenerate \cite[Sec.~V.A]{DFT-graphs}, so all our examples will be for $N=2$ non-interacting spinless particles, where orbitals from the first-excited state also get filled.
Using the orbital expressions we get
\begin{equation}\label{eq:P-factors-2-particles}
    \oirho{k}{i} = \phi_{0,i}^2 + \phi_{k,i}^2, \quad
    \oirho{kl}{i} = 2\phi_{k,i}\phi_{l,i},
\end{equation}
where $k,l \in \{1,2,3\}$ and $i$ ranges over the lattice sites, $i\in\{1,\ldots,M\}$. Putting in the numbers this gives
\begin{align}
    \orho{k} &= \overline{\rho} = \tfrac{1}{2}(1,1,1,1)\\
    \orho{12} &= \tfrac{1}{2} (1,-1,-1,1)\\ \orho{13} &= \tfrac{1}{2} (-1,1,-1,1) \\ \orho{23} &= \tfrac{1}{2} (-1,-1,1,1).
\end{align}
These vectors now make up the six columns in $P$ and we can check that indeed $\kappa=\dim\ker P = 2$. For $D_\R$ as the image of all $x\in\R^3$, $\|x\|=1$, under $\rho$ we get
\begin{equation}\label{eq:D_R-Roman-surface}
    \rho(x) = \overline{\rho} \underbrace{\sum_{k=1}^3 x_k^2}_{1} + \frac{1}{2}\begin{pmatrix*}[r]
        1 & -1 & -1 & 1 \\
        -1 & 1 & -1 & 1 \\
        -1 & -1 & 1 & 1 \\
        1 & 1 & 1 & 1
    \end{pmatrix*}
    \begin{pmatrix*}
        x_1x_2\\
        x_1x_3\\
        x_2x_3\\
        0
    \end{pmatrix*}.
\end{equation}
The surface defined by the last vector and parametrized by the unit sphere $\|x\|=1$ is an amazing structure of tetrahedronal symmetry that is known as \emph{Roman surface}, while the matrix in front is orthogonal and just describes a mirroring at a plane with normal vector $(1,1,1,-1)$. The situation is displayed in Figure~\ref{fig:Roman-surf}. The full density region $D$ is then the convex hull of $D_\R$ and in this case is equal to the segment set, so we have $D = D_\C = \seg D_\R = \ch D_\R$.
\begin{figure}[ht]
    \centering
    \includegraphics[width=.9\columnwidth]{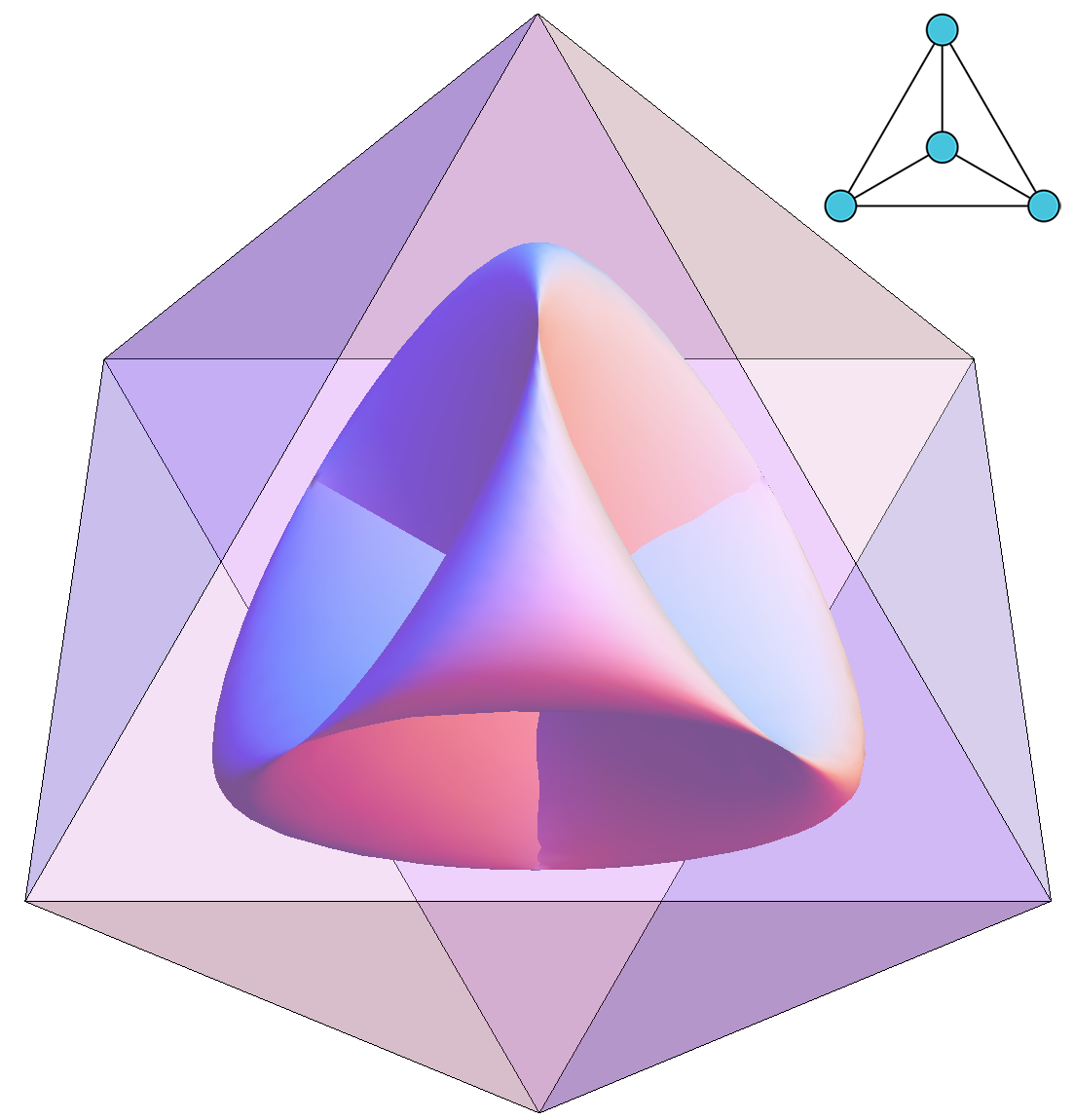}
    \caption{The three-fold degeneracy region $D_\R$ of the tetrahedron graph inside the octahedronal density domain. The corners of the octahedron correspond to the extreme density $(1,1,0,0)$ and its permutations, generalized barycentric coordinates are used to display densities.}
    \label{fig:Roman-surf}
\end{figure}
Since $P$ maps $\R^6\to \R^4$ and $\kappa = \dim\ker P = 2$, an embedding of $\R^4$ into $\R^6$ orthogonal to $\ker P$ allows to see $P$ as a projection with projection center $\ker P$ \cite[1.2.7]{fortuna2016book}. Using the following general lemma, that the projection center does not intersect the Veronese variety, helps to reduce the number of possible classes of Steiner surfaces appearing as density regions next to the Roman surface.

\begin{lemma}\label{lem:VKerPzero}
$\mathbb{V}_g \cap \ker P = \emptyset$.
\end{lemma}

\begin{proof}
This follows directly from the fact that for all $x\in\R^g$ with $\|x\| = 1$, as the parameters of the Veronese variety, $\rho(x)=P\circ \nu(x)$ is a normalized density, so it cannot be zero.
\end{proof}

A subclassification of Steiner surfaces with $(g,\kappa) = (3,2)$ conducted by \citet{degen1994} considers various positions of the projection center and basically looks at the shadows cast by the ellipses that compose the Veronese variety.
By Theorem~4.1 from the reference and our Lemma~\ref{lem:VKerPzero} their class (B) is ruled out.
If an ellipse is projected such that only a line remains, this appears as a singular point in projective space, and counting those points gives the five possible subclasses (Aa)-(Ae). In the case (Aa), three ellipses from the Veronese surface are projected exactly such that they appear as lines, the three lines of self-intersection of the Roman surface. The shapes of class (C) are quadrics, such as ellipsoids, that did not show up in the examples studied here. If Lemma~\ref{lem:VKerPzero} could be strengthened to state that the complex extension of the Veronese variety has empty intersection with the projection center then also class (C) can be excluded.
Since the possible surfaces that one can construct like this correspond to triangular Bézier-surface patches, this classification also has a relevance in computer graphics and finite element methods \cite{sederberg1985steiner}. A similar type of classification directly on the basis of the matrix $P$ that is less geometrical was given by \citet{Coffman-Roman-surf} and \citet[Ch.~1.3]{apery-book}, while the historically first classification was seemingly put forward in the book of \citet[Ch.~XV]{Michel-book} based on pencils of conics.
For some classical literature that presents more interesting properties in particular of the Roman surface, see \citet{clebsch1867}, \citet{cayley1873steiner}, \citet{lacour1898} and \citet[§46]{hilbert-geometry-book}.
Very recently, the possibility of physically realizing the Roman surface by a spin-induced polarization vector in a particular cubic crystal was reported in \citet{liu2022physical-Roman-surf}.
The example of the Roman surface as the degeneracy region for $v=0$ in a tetrahedron graph clearly shows that while an individual ground-state density does not need to reflect the symmetry of the system, the whole degeneracy region does.

After this discovery, we expect even more intricate geometric objects for $g=3,\kappa<2$ or $g>3$. However, it is clear from the fact that a higher-dimensional $\mathcal{U}$ is composed of lower-dimensional subspaces that all $D_\R$ are parametrized families of other $D_\R$ of lower degree, just like the Roman surface can be constructed from a family of ellipses (see Figure~\ref{fig:Roman-surf-Ellipse}). Note that in the case $\kappa=0$ the full Veronese variety is retained as a density region since all factors in \eqref{eq:P-nu-map} are linearly independent. This actually happens in the cuboctrahedron-graph example presented in in Sec.~VI.E of \citet{DFT-graphs} that is also discussed in Section~\ref{sec:non-pure-v-rep} here.
For general $(g,\kappa)$ a further subclassification, like it was conducted by \citet{degen1994} for $(g,\kappa)=(3,2)$ and which was briefly discussed above, seems entirely possible.
\begin{figure}[ht]
\begin{subfigure}{.45\columnwidth}
  \centering
  \includegraphics[width=\columnwidth]{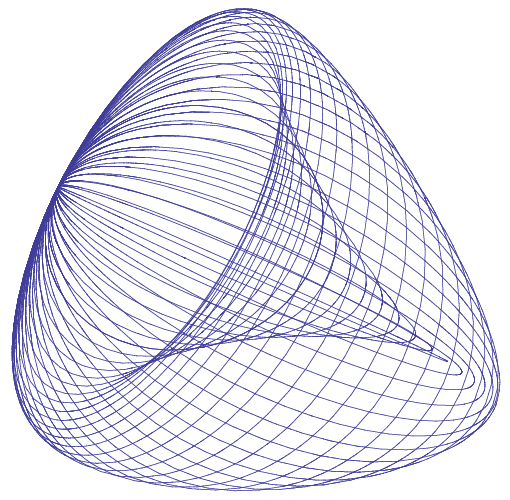}
\end{subfigure}
\hfill 
\begin{subfigure}{.5\columnwidth}
  \centering
  \includegraphics[width=\columnwidth]{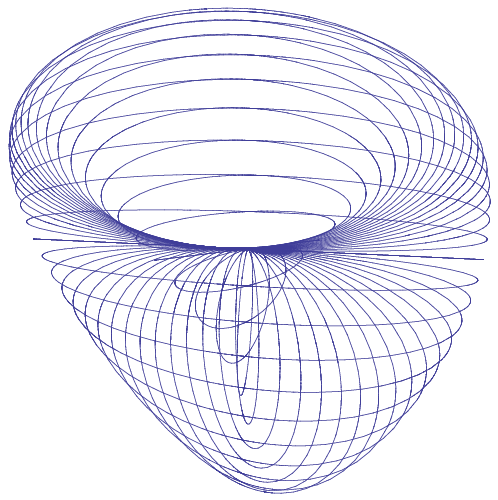}
\end{subfigure}
\hfill 
\begin{subfigure}{1\columnwidth}
  \centering
  \begin{subfigure}{.45\columnwidth}
  \includegraphics[width=\columnwidth]{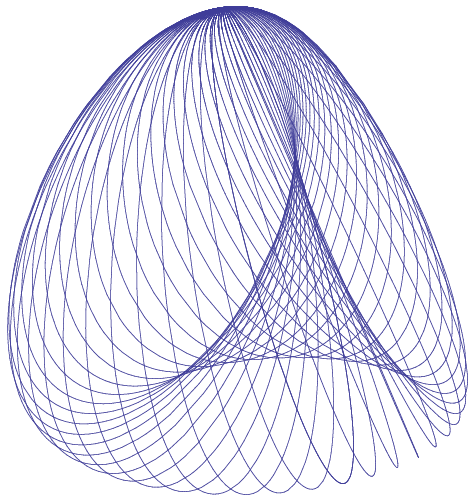}
  \end{subfigure}
\end{subfigure}
\caption{Three different parametrizations for the Roman surface that all show how it can be constructed from ellipses. The first corresponds to introducing polar coordinates for $x\in\R^3$ in \eqref{eq:P-nu-map}, the second to doing the same in \eqref{eq:D_R-Roman-surface}, while the third is from \citet[Sec.~2.4]{apery-book}.}
\label{fig:Roman-surf-Ellipse}
\end{figure}

\section{Non-pure-state $v$-representable densities}
\label{sec:non-pure-v-rep}

As already exemplified for the cases $g=2$ and $g=3$ in the previous section, the $g$-dimensional $\mathcal{U}$ can be thought of as the $g$-fold degenerate eigenspace for a given Hamiltonian with potential $v$. Since the Hamiltonian is assumed real symmetric, Lemma~\ref{lem:real-eigenstate} shows that such an eigenspace can always be spanned by real eigenvectors. Then we write $D(v)$ for the corresponding density region and it is called a \emph{degeneracy region}, the set of all possible densities coming from (ensemble) ground states to the potential $v$. As demonstrated in Section~\ref{sec:subspaces-density-map}, the degeneracy regions can be classified by $(g,\kappa)$. 
We define the \emph{subdifferential} of a convex functional at a certain point as the set of all tangent lines in that point that lie fully below the functional \cite[Def.~2.30]{Barbu-Precupanu}. Likewise the \emph{superdifferential} of a concave functional collects all tangent lines above.
Then, interestingly, it holds that $D(v) = \overline\partial E(v)$, the degeneracy region is the superdifferential of the concave energy functional $E(v) = \inf_\Psi \{\langle \Psi, H_v\Psi \rangle\}$ that itself is the convex conjugate of the universal functional $F(\rho)$. On the other hand, $v\in -\underline\partial F(\rho)$, the potential lies in the subdifferential of the universal functional if $\rho$ is the ground-state density for potential $v$.
\cite{Kvaal2014,laestadius2018generalized}

From the definition of $D_\C(v)$ it follows that this set includes all densities coming from pure ground states (pure-state $v$-representable). It is then known from arguments in \citet{Levy1982} and \citet{Lieb1983} that if $g\geq 3$ there can occur densities that are \emph{not} pure-state $v$-representable, i.e., they are in $D(v)$ but not in $D_\C(v)$, but the provided proofs remained vague (in \cite[Sec.~VI.D]{DFT-graphs} we even find a counterexample to a statement in the proof of \citet[Th.~3.4(i)]{Lieb1983}). The following theorem helps to establish a precise condition when and for which density non-pure-state $v$-representability will occur. More precisely, the condition $\kappa=0$ exactly prevents ``certain linear dependencies'' that are mentioned in \citet{Levy1982} as a situation where pure-state $v$-representability is still possible.

\begin{lemma}\label{lem:central-density-not-D_C}
If $g\geq 3$ and $\kappa=0$ then the central density $\overline{\rho} = (\orho{1}+\ldots+\orho{g})/g$ is \emph{not} in $D_\C$.
\end{lemma}

\begin{proof}
Because of Lemma~\ref{lem:DC=segDR} in order for $\overline{\rho}$ to be in $D_\C$ it must hold that there is a $\lambda\in[0,1]$ and $x,y\in \R^g$ with $\|x\|=\|y\|=1$ such that
\begin{equation}\label{eq:central-dens-convex}
    (\orho{1}+\ldots+\orho{g})/g = \lambda \rho(x) + (1-\lambda) \rho(y),
\end{equation}
with $\rho(x),\rho(y) \in D_\R$ given by \eqref{eq:P-nu-map}.
Since $\kappa=0$ it holds $\ker P=\{0\}$ and we can apply $P^{-1}$ and formulate the equation above on the codomain of the Veronese embedding as
\begin{align}
    \frac{1}{g} &= \lambda x_k^2 + (1-\lambda) y_k^2 \label{eq:veronese-map-range-1} \\
    0 &= \lambda x_kx_l + (1-\lambda) y_ky_l \quad\quad (k\neq l) \label{eq:veronese-map-range-2}
\end{align}
with $k,l \in \{1,\ldots,g\}$.
Suppose that some $x_k=0$ then $\lambda\neq 1$ and $y_k\neq 0$ follows from \eqref{eq:veronese-map-range-1} and for all $l\neq k$ it holds $y_l=0$ from \eqref{eq:veronese-map-range-2}. Then \eqref{eq:veronese-map-range-1} again says that all $x_l, l\neq k,$ have the same non-zero value and that $\lambda\neq 0$. Now since $g\geq 3$ we can choose $l\neq m$ both different from $k$ and then \eqref{eq:veronese-map-range-2} yields a contradiction. The same holds if some $y_k=0$. Therefore, it will be assumed from here on that all coefficients $x_k,y_k \neq 0$. This implies $\lambda \notin \{0,1\}$ because else \eqref{eq:veronese-map-range-2} cannot be fulfilled. We can thus write
\begin{equation}
    x_k x_l = - \frac{1-\lambda}{\lambda} y_k y_l = - \mu \, y_k y_l \quad \quad (k \neq l),
\end{equation} 
where $\mu =  (1-\lambda)/\lambda >0$. This gives for $k,l,m$ all different,
\begin{equation}
\begin{aligned}
&(x_k x_l) (x_k x_m) = \mu^2 (y_k y_l) (y_k y_m) \\
\Rightarrow\; &x_k^2 (x_l x_m) = \mu^2 y_k^2 (y_l y_m) \\
\Rightarrow\; &-\mu \, x_k^2 (y_l y_m) =  \mu^2 y_k^2 (y_l y_m) \\ \Rightarrow\; &x_k^2 =- \mu \, y_k^2.
\end{aligned}
\end{equation}
However, this is a contradiction since $\mu > 0$ and $x_k,y_k$ were assumed to be non-zero and it follows that $\overline{\rho} \notin D_\C$.
\end{proof}

\begin{corollary}\label{cor:non-pure-state-v-rep}
If there exists a potential $v$ that gives a degeneracy region $D(v)$ with $g\geq 3$ and $\kappa=0$ then the central density $\overline{\rho}$ plus a neighbourhood is non-pure-state $v$-representable.
\end{corollary}

\begin{proof}
This is a direct consequence of Theorem~\ref{th:HK1}. If $\overline{\rho} \in D(v)$ would be the density of a pure ground state $\Psi$ of any $v'$ then this $\Psi$ must also be the (pure) ground state for $v$ which is in contradiction to the statement of Lemma~\ref{lem:central-density-not-D_C} that $\overline{\rho} \notin D_\C(v)$. Since $D_\C(v)$ was noted to be closed, we can find an open set including $\overline{\rho}$ that has an empty intersection with $D_\C(v)$.
\end{proof}

An explicit example for non-pure-state $v$-representability was given by the authors in Sec.~VI.E of \citet{DFT-graphs} with two non-interacting spinless particles on a cuboctahedron-graph system that exhibits three-fold degeneracy in the ground state. We will briefly show that this example fits to the results from above, i.e., it yields a $P$ with kernel zero. In order to determine the $12\times 6$ matrix $P$ (12 because of the number of vertices in the cuboctahedron and 6 from the dimension of the codomain of the Veronese embedding for $g=3$) we have to calculate the factors $\orho{k}, \orho{kl}$ from the one-particle orbitals $\phi_k$ where the orbital energies are $\epsilon_0 < \epsilon_1=\epsilon_2=\epsilon_3$, all written out in the reference. Like before in the tetrahedron-graph example, these orbitals yield the two-particle Slater determinants $\Phi_k = \phi_0 \wedge \phi_k$, $k\in \{1,2,3\}$, that span the three-dimensional ground-state subspace. Equation~\eqref{eq:P-factors-2-particles} then again determines the vectors that form the columns of $P$ which can be checked to give a matrix with $\ker P=\{0\}$. Similar investigations show that the icosahedron and dodecahedron graphs for two particles described by graph Laplacians have $(g,\kappa)=(3,0)$ at $v=0$ too.

\section{Degeneracy-preserving potential mani\-folds}
\label{sec:Ullrich-Kohn}

When degeneracy occurs for a fixed $H_0$ and a certain potential $v$, variation of this potential can still preserve the symmetry and thus the degree of degeneracy. For finite-lattice systems, the work of \citet{ullrich2002} gives an upper limit to the dimensional size of the manifold in which such potentials with $g$-fold degeneracy are contained.
The vertices of the lattice will be indexed by $i \in \{1, \ldots, M \}$ and thus a density as well as a scalar potential are given as vectors $\rho,v\in\R^M$. Since the density $\rho_i$ at each vertex $i$ must be in $[0,1]$ for fermionic particles and since the whole density is normalized to the number of particles, $\sum_i\rho_i = N$, the density domain has the shape of an $(M,N)$-hypersimplex $\mathcal{P}_{M,N}$.
The dual pairing (inner product) between the potential and the density that yields the potential energy is $\langle v,\rho \rangle = \sum_i v_i\rho_i$. The whole expectation value of the Hamiltonian for a normalized state $\Psi\in\H$ with density $\rho$ is thus $\langle \Psi,(H_0+V)\Psi \rangle = \langle \Psi,H_0\Psi \rangle + \langle v,\rho \rangle$.
We will extend the main result of \citet{ullrich2002} with a proof resting solely on the classification for degeneracy regions introduced in Section~\ref{sec:subspaces-density-map} and an extremely useful result of \citet{rellich1937,rellich-book} (see also \citet{Kato-book} for a more modern treatment) on the analyticity of eigenvalues and eigenvectors under perturbations.

\begin{theorem}\label{th:Ullrich-Kohn}
Let $H_0$ be a real symmetric matrix and $v \in \R^M$ a potential such that $H_v=H_0+V$ has a degeneracy region of class $(g,\kappa)$. Then $v$ (modulo constants) is contained in a potential manifold of at most dimension $M-(g+1)g/2+\kappa$ in which the degree of degeneracy $g$ is preserved. 
\end{theorem}

\begin{proof}
Let $H_v$ have a $g$-fold degenerate ground state. Then in the ground-state eigenspace we can choose the real orthonormal eigenstates $\Phi_1, \ldots, \Phi_g$ by Lemma~\ref{lem:real-eigenstate}.
Let $U$ be a perturbation of the potential with one-body potential $u$ such that for the perturbed Hamiltonian $H_v(\lambda)=H_v + \lambda U$ the $g$-fold ground-state degeneracy is preserved. As usual, adding a constant potential to $U$ does not change this situation. According to Rellich's theorem we can choose $g$ orthonormal ground states $\Psi_1 (\lambda),\ldots, \Psi_g (\lambda)$ with common eigenvalue $E(\lambda)$ which are analytic in $\lambda$ for small enough $|\lambda|$. These states are generally not analytically connected to our chosen basis $\{\Phi_k\}_{k=1}^g$ at $\lambda=0$.
However, the states
\begin{equation}
    \tilde{\Phi}_k (\lambda) = \sum_{j=1}^g \Psi_j (\lambda) \langle \Psi_j (\lambda), \Phi_k \rangle 
\end{equation}
are analytic in $\lambda$ as well and do have the property that $\tilde{\Phi}_k (0)=\Phi_k$, which follows by noting that the states $\Psi_j(0)$ form a complete basis set in the ground-state eigenspace of the unperturbed system.
Using the eigenstate property $H_v(\lambda)\tilde\Phi_k(\lambda) = E(\lambda)\tilde\Phi_k(\lambda)$ we differentiate with respect to $\lambda$ and then set $\lambda=0$ to get
\begin{equation}
    H_v \tilde\Phi_k'(0) +  U \Phi_k = E (0) \tilde\Phi_k'(0) + E'(0) \Phi_k.
\end{equation}
We project this onto the subspace spanned by $\Phi_l$ and use $H_v \Phi_l=E(0) \Phi_l$ to get
\begin{equation}
\langle \Phi_l, U \Phi_k \rangle = E'(0) \delta_{kl}.
\end{equation}
Since we have the possibility to add any constant potential to $U$, we can use this gauge freedom to set $E'(0)=0$ and thus arrive at the conditions $\langle \Phi_1, U \Phi_1 \rangle = \ldots = \langle \Phi_g, U \Phi_g \rangle = 0$ in the diagonal and $\langle \Phi_l, U \Phi_k \rangle = 0$ off-diagonal, a total of $g + g(g-1)/2 = (g+1)g/2$ constraints. However, by remembering the definition of $\orho{k}, \orho{kl}$ from Section~\ref{sec:subspaces-density-map}, we see that those constraints are actually equivalent to $\langle u,\orho{k} \rangle = 0$ and $\langle u,\orho{kl} \rangle = 0$ for $k,l\in\{1,\ldots,g\}$, $k<l$, which is nothing but the matrix equation $\T{P} u = 0$. Thus, if $u$ is in the kernel of the transpose of the $(M \times (g+1)g/2)$-matrix $P$, the constraints are fulfilled and this assures that the added potential does not change the degree of degeneracy to first order in $\lambda$ (since it still can change the degree to second order, we only get an upper bound for the dimension of the invariant manifold). To get the dimension of the kernel, we simply use the rank-nullity theorem for $P$ and $\T{P}$ and the fact that the rank of a matrix is equal to the rank of its transpose.
\begin{align}
    &\dim\ker P = \kappa = \frac{1}{2}(g+1)g - \rank P \\
    &\dim\ker \T{P} = M - \rank P = M - \frac{1}{2}(g+1)g + \kappa
\end{align}
These considerations restrict the dimension of the potential manifold (modulo constants) that preserves the $g$-fold degeneracy to at most $M-(g+1)g/2+\kappa$.
\end{proof}

The previous result of \citet{ullrich2002} did not contain a reference to the nullity and implicitly assumed $\kappa=0$, so their dimension is $M - (g+1)g/2$, which is less then the possible maximal dimension of Theorem~\ref{th:Ullrich-Kohn} if $\kappa>0$. Taking the maximal nullity, $\kappa=(g+1)g/2-1$, the maximal dimension for the potential manifold that preserves degeneracy evaluates as $M-1$, which is just the full dimensionality of the potential space (modulo constants) and thus invalidates the argument of \citet{ullrich2002} that degeneracy is rare in potential space. Nevertheless, the potentials that lead to degeneracy are still expected to be of measure zero in the absence of internal degrees of freedom, since they either have to obey certain symmetries or trigger accidental degeneracy, so Theorem~\ref{th:Ullrich-Kohn} is simply not useful for the case of maximal nullity. The result also does not fit to the maximal dimension derived in \citet[Sec.~IV.B]{DFT-graphs} from the \emph{original} result of \citet{ullrich2002} as $M-3$, which lead to the premature conclusion that the set of non-uv potentials has measure zero.

The relevance of Theorem~\ref{th:Ullrich-Kohn} in the present context is that whenever $v$ leads to a degeneracy region $D(v)$ of class $(g,\kappa)$, it can actually be the member of a whole family of degeneracy regions of the same degree with an $(M-(g+1)g/2+\kappa)$-dimensional parametrization. Each degeneracy region within this family can be reached by moving the potential inside the manifold that preserves the degree of degeneracy. Such a family will be called a \emph{degeneracy bundle}. Note that the separate degeneracy regions in such a bundle can lie arbitrarily close together but never touch, a notion that will become important in Section~\ref{sec:geom}.
Since the dimension of a degeneracy region is given by \eqref{eq:max-dim-D}, the largest possible dimension of a whole degeneracy bundle is
\begin{equation}
    M-\frac{1}{2}(g+1)g+\kappa + \frac{1}{2}(g+1)g-\kappa-1 = M-1,
\end{equation}
which is precisely the full dimensionality of the density domain when normalization is taken into account. This led \citet{ullrich2002} to the statement that degeneracy can be considered common in the density domain. Precise numbers for the ratio of degeneracy in the density domain are calculated for a few examples in the concluding Section~\ref{sec:summary}.

To give an example, we will resort again to the tetrahedron graph with two particles from Section~\ref{sec:density-regions}. There we studied the degeneracy set $D(0)$ of class $(g,\kappa)=(3,2)$, the convex hull of the Roman surface $D_\R(0)$, that appears for $v=0$, the potential with maximal symmetry. Here Theorem~\ref{th:Ullrich-Kohn} for $M=4$, $g=3$, $\kappa=2$ leaves no possible dimension for the manifold of invariant degeneracy.
Indeed, by changing the potential, we will necessarily diminish the system of symmetry and thus reduce the degree of degeneracy. However, for a potential of type $v=(s,0,0,0)$ (and all permutations), $s>0$, still a degeneracy of class $(g,\kappa)=(2,0)$ remains, since the vertices apart from the first one still retain a triangular symmetry. (For $s<0$ the two-fold degeneracy is in the first excited state.)
The application of Theorem~\ref{th:Ullrich-Kohn} for $M=4$, $g=2$, $\kappa=0$ now gives a one-dimensional manifold, parametrized by $s>0$, in which the degree of degeneracy is preserved.
In the limit $s\to \infty$ the first vertex is effectively removed from the system and the case of the triangle graph, with the incircle of the triangular density domain $\mathcal{P}_{3,2}$ as degeneracy region, remains. In summary, four degeneracy bundles in the shape of cylinders parametrized by $s>0$ will extend from the flat sides of the convexified Roman surface until they reach the faces of the octahedron $\mathcal{P}_{4,2}$. The whole situation is displayed in Figure~\ref{fig:Roman-surf-with-cylinders}. The collection of all degeneracy regions and bundles will be called the \emph{degeneracy structure}. In the next section we investigate how this geometrically relates to non-uv densities and the corresponding potentials.
\begin{figure}[ht]
    \centering
    \includegraphics[width=.9\columnwidth]{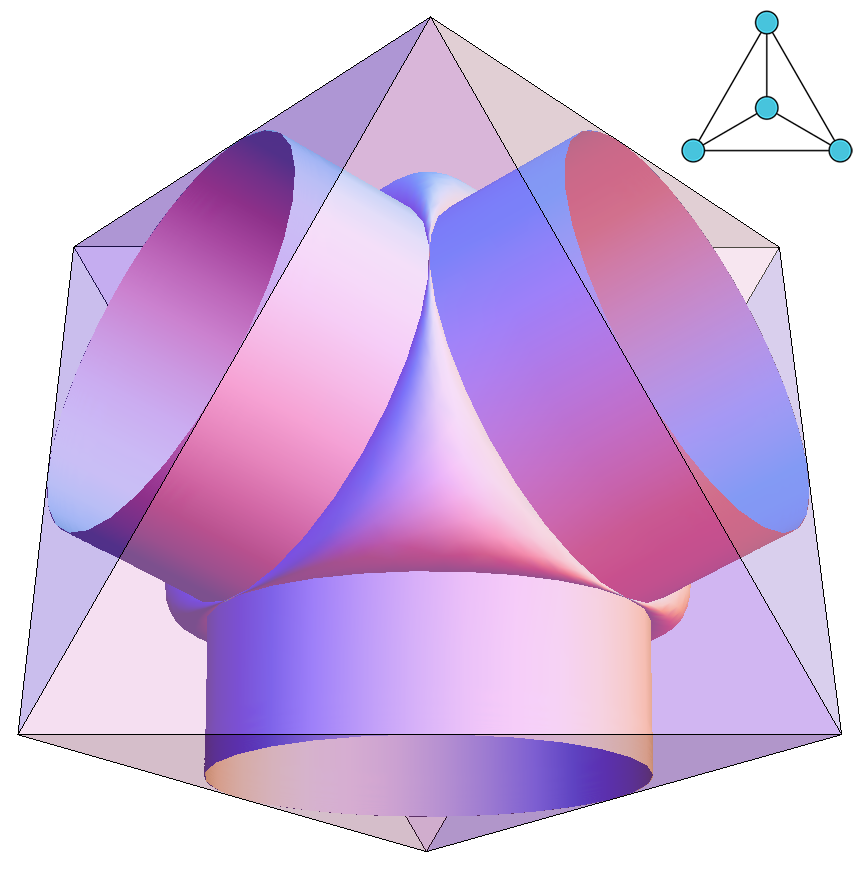}
    \caption{The union of all degeneracy regions (\emph{degeneracy structure}) of the tetrahedron graph inside the octahedral set of densities $\mathcal{P}_{4,2}$. The four degeneracy bundles in the shape of cylinders reaching out from the flat sides of the convexified Roman surface have to be imagined as filled.}
    \label{fig:Roman-surf-with-cylinders}
\end{figure}

\section{Geometrization of the potential-density mapping}
\label{sec:geom}

Contrary to the continuum case of standard DFT, where the Hohenberg--Kohn theorem implies that every ground-state density comes from a unique potential, at least if those are limited to a certain class \cite{Garrigue2018}, the lattice case introduced in Section~\ref{sec:Ullrich-Kohn} allows for non-uniquely $v$-representable (``non-uv'') densities. We proved their existence together with some examples and basic properties in our previous work \cite{DFT-graphs}, where we also put forward the idea of a full geometrization of the mapping from potentials to densities with degeneracy regions and non-uv densities and their corresponding potentials as the prime elements. We wrongly conjectured that degeneracy regions have a spherical shape, while actually their shape is much more intriguing, as was demonstrated in Section~\ref{sec:density-regions}. With these first results in hand, we will here finalize and prove the geometry conjecture.


\begin{theorem}[Density-Potential Geometry]
\label{th:geometry}
The occurrence of non-uv densities in the given setting is limited to two situations:\begin{enumerate}[(a)]
\item If two different degeneracy regions $D(v_\mathrm{I}),D(v_\mathrm{II})$ intersect (in a single density or forming a degeneracy region of strictly lower degree) then $\rho\in D(v_\mathrm{I})\cap D(v_\mathrm{II})$ is non-uv. All potentials on the line segment that connects $v_\mathrm{I}$ and $v_\mathrm{II}$ yield the same ground-state density $\rho$.
\item If a degeneracy region $D(\vA)$ touches the boundary of the density domain $\mathcal{P}_{M,N}$ then this boundary density $\rho$ is non-uv. All potentials from a ray with apex $\vA$ yield the same ground-state density $\rho$. Furthermore, such densities are the \emph{only} $v$-representable densities on the boundary of $\mathcal{P}_{M,N}$, other boundary points are \emph{not} $v$-representable.
\end{enumerate}
\end{theorem}

\begin{proof}
We begin with a general observation concerning non-uv densities.
If a density $\rho$ arises from two different potentials $v,v'$ then by Theorem~\ref{th:HK1} both Hamiltonians $H_v,H_{v'}$ share a ground-state density matrix $\Gamma\mapsto\rho$. It is then clear that this $\Gamma$ is also an eigenstate of $H_{v_\lambda}$ with $v_\lambda=\lambda v+(1-\lambda) v'$ for any $\lambda \in \R$.
But since $\rho$ is the ground-state density for both potentials $v,v'$, it must also minimize both convex functionals $F(\rho)+\langle v,\rho \rangle$, $F(\rho)+\langle v',\rho \rangle$ by the definition of a ground state. It is then also a minimizer for $\lambda(F(\rho)+\langle v,\rho \rangle)$, $(1-\lambda)(F(\rho)+\langle v',\rho \rangle)$ ($\lambda \in [0,1]$) and their sum $F(\rho)+\langle \lambda v+(1-\lambda) v',\rho \rangle$. This shows that $\rho$ is the \emph{ground-state} density for \emph{all} potentials on the line segment between $v$ and $v'$.
Since the eigenvalues of $H_{v_\lambda}$ are all continuous in $\lambda$ by Rellich's theorem, there are two possibilities if $\lambda$ is varied into both directions outside the interval $[0,1]$: Either level crossings occur and another ground state takes over in both directions (Fig.~\ref{fig:2-crossings}), or $\Gamma$ remains the ground state while $|\lambda|\to \infty$ into one direction and a level crossing occurs in the other direction (Fig.~\ref{fig:1-crossing}). We will link those two possible situations to the statements (a) and (b) of the theorem in the following.
\begin{figure}[ht]
  \begin{subfigure}{.45\columnwidth}
  \centering
    \begin{tikzpicture}[scale=0.45]
        \draw[->] (0, 0) -- (6.5, 0) node[below] {$\lambda$};
        \draw[->] (0, 0) -- (0, 4) node[left] {$E$};
        \draw[-, MidnightBlue] (0, 2.5) -- (6, 2.5);
        \draw[domain=0:6, smooth, variable=\x, MidnightBlue] plot ({\x}, {-1/3*(\x-3)*(\x-3)+3.25});
        \node[below] at (1.5,0) {$\lambda_1$};
        \node[below] at (4.5,0) {$\lambda_2$};
        \draw[-,dashed] (1.5, 0) -- (1.5, 2.5);
        \draw[-,dashed] (4.5, 0) -- (4.5, 2.5);
        \fill[MidnightBlue] (1.5, 2.5) circle (3pt);
        \fill[MidnightBlue] (4.5, 2.5) circle (3pt);
    \end{tikzpicture}
    \caption{}
    \label{fig:2-crossings}
  \end{subfigure}
  \hfill 
  \begin{subfigure}{.45\columnwidth}
    \centering
    \begin{tikzpicture}[scale=0.45]
        \draw[->] (0, 0) -- (6.5, 0) node[below] {$\lambda$};
        \draw[->] (0, 0) -- (0, 4) node[left] {$E$};
        \draw[-, MidnightBlue] (0, 2) -- (6, 2);
        \draw[-, MidnightBlue] (0, 0.5) -- (6, 3.5);
        \draw[-,dashed] (3, 0) -- (3, 2);
        \fill[MidnightBlue] (3, 2) circle (3pt);
        \draw[->, thick] (2.2, -.5) to (2.9, -.1);
        \node[below,left] at (2.3, -.6) {$\vA$};
    \end{tikzpicture}
    \caption{}
    \label{fig:1-crossing}
  \end{subfigure}
  \caption{The two situations of level crossings in the proof of Theorem~\ref{th:geometry}.}
\end{figure}
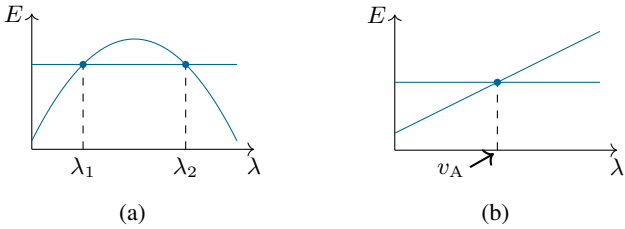
The case that it remains the ground state in both directions is impossible, as the following will show.\\
To show statement (a), we take $\rho\in D(v_\lambda)$, where in this special case $D(v_\lambda)$ can also denote a set containing the single ground-state density for a potential $v_\lambda$ if no degeneracy arises.
If level crossings occur in both directions, say at $\lambda_1$ and $\lambda_2$ with $\lambda_1<\lambda_2$, this implies additional degeneracy at $v_\mathrm{I}=v_{\lambda_1}$ and $v_\mathrm{II}=v_{\lambda_2}$. Any non-uv density $\rho \in D(v_\lambda)$, $\lambda_1<\lambda<\lambda_2$, then belongs to two degeneracy regions $D(v_\mathrm{I})$, $D(v_\mathrm{II})$ of strictly higher degree than $D(v_\lambda)$ because the crossings add at least one to the degree. This means $D(v_\lambda) = D(v_\mathrm{I}) \cap D(v_\mathrm{II})$, so $\rho$ really lies in the intersection of two degeneracy regions.\\
For statement (b) consider the case where $|\lambda|\to\infty$ is possible in one direction without a level crossing. This means at least one component $v_{\lambda,i}$ of $v_\lambda$ diverges.
If $v_{\lambda,i} \to +\infty$ then it must hold that $\rho_i = 0$ for this density to still belong to the ground state. But then performing $v_{\lambda,i} \to -\infty$ means that at one point another eigenstate (and there must always be at least one vertex $i$ with $\rho_i>0$) will take over the role of the ground state. But this must happen at a level crossing, say at $\vA$, where the energies match (because by Rellich's theorem the variation of $v$ leads to a continuously changing energy) and such we arrive at a degeneracy region $D(\vA) \ni \rho$.
If $v_{\lambda,i} \to -\infty$ then since any ground-state density minimizes $F(\rho)+\langle v_\lambda,\rho \rangle$ and $F$ is bounded on $\mathcal{P}_{M,N}$ (as a finite convex function on a compact domain), only a density with $\rho_i=1$ qualifies. Going into the opposite direction $v_{\lambda,i} \to +\infty$, the energy must increase and thus at one point another eigenstate will again take over.
Consequently, if $|\lambda|\to \infty$ is possible in one direction without a level crossing for a non-uv density, this density simultaneously belongs to the boundary of $\mathcal{P}_{M,N}$ (because it has either $\rho_i=0$ or $\rho_i=1$) and to a degeneracy region $D(\vA)$.\\
Conversely, assume a $v$-representable $\rho$ is on the boundary of $\mathcal{P}_{M,N}$, which means it has at least one component $\rho_i=0$ or $\rho_i=1$. In the case $\rho_i=1$, $v_i$ can always decrease from its previous value and will only make the energy smaller, also compared to the energy of other eigenstates that have $\rho_i\leq 1$, thus not influencing the ground state. This gives a ray of potentials that all lead to the same density $\rho$, thus making it non-uv. Yet, by continuously increasing $v_i$ one will find the end point, call it $\vA$ as before, of such a ray when the potential is large enough such that another eigenstate takes over and assumes the role of ground state, exactly at a point of degeneracy.
In the case $\rho_i=0$, the ray just extends into the other direction, since changing $v_i$ arbitrarily will not influence the energy of the original ground state, but continuously decreasing $v_i$ will eventually allow another eigenstate to take over the role of the ground state, again at a point of degeneracy. This proves that a density on the boundary of $\mathcal{P}_{M,N}$ is $v$-representable if and only if it also belongs to a degeneracy region $D(\vA)$.
\end{proof}

It must be noted that the last part of the proof shows that if for a $v$-representable boundary density multiple components $\rho_i$ are either 0 or 1, which is possible at least in principle, the rays combine into a convex set of non-uv potentials. But such a situation was never observed in any of the examples studied in the course of this work, just as the possibility of degeneracy regions touching each other or the boundary in more than a single point has not been verified. While singleton degeneracy regions are not principally ruled out in the theorem, degeneracy regions that just arise from internal degrees of freedom do not play any role in the context of the theorem, because the external potential does not couple to the internal coordinates.

The following corollary, put forward as a conjecture in \citet{DFT-graphs}, is imminent from the fact that degeneracy regions are always closed and that non-uv densities only arise as their intersection or when they touch the boundary. We also include our previous results that the non-uv densities have measure zero in the density domain \cite[Sec.~III.D]{DFT-graphs}, which is visible from the density-potential geometry theorem as well.

\begin{corollary}
The set of non-uv densities is closed and has measure zero in the $(M-1)$-dimensional density domain $\mathcal{P}_{M,N}$.
\end{corollary}

Conversely, this shows that almost all densities are uniquely $v$-representable, which means the failure of the Hohenberg--Kohn theorem for lattice systems can be considered less drastic. If a density is non-uv after all then by Theorem~\ref{th:geometry} it must lie at the touching point of two degeneracy regions or where a degeneracy region touches the border of the density domain, a quite peculiar property.
If, on the other hand, we can guarantee the validity of the Hohenberg--Kohn theorem in some setting, then this in turn means that degeneracy regions can never touch each other or the boundary.
Another previous result \cite[Cor.~20]{DFT-graphs} is also a direct consequence of the theorem on density-potential geometry: The set of potentials leading to a non-uv density that is not on the boundary of the density domain is always bounded.

We would like to give examples for the geometrical situations present in Theorem~\ref{th:geometry}, two where a degeneracy region touches the boundary of the density domain $\mathcal{P}_{M,N}$ and another where degeneracy regions touch each other. The first situation was already present in the triangle graph example discussed in Sec.~VI.C of \citet{DFT-graphs} where the degeneracy region for $v=0$ touches the boundary of $\mathcal{P}_{3,2}$ at the three points where the incircle touches an equilateral triangle and correspondingly three rays of non-uv potentials shoot off from $v=0$ to infinity in the potential space. Another example can be found in the tetrahedron graph that features the Roman surface as a degeneracy region for $v=0$ and that has already been discussed twice here. The Roman surface touches the boundary of the octahedron $\mathcal{P}_{4,2}$ exactly at the four points $(1,\tfrac{1}{3},\tfrac{1}{3},\tfrac{1}{3})$ (and permutations) corresponding to the vertices of its tetrahedral symmetry. Those are opposite to the flat face where the cylinder-shaped degeneracy bundles are attached, as can be seen from Figure~\ref{fig:Roman-surf-with-cylinders}. Since those degeneracy bundles correspond to the potential $v=(s,0,0,0)$ (and all permutations), $s>0$, one can intuitively state that the non-uv potential rays corresponding to the touching points are given by $v=(s,0,0,0)$ (and all permutations), $s<0$. On the other hand, the tetrahedron graph example, although it contains a large amount of degeneracy, does not feature any degeneracy regions touching each other.
For this we switch over to two non-interacting spinless particles on a simple square graph, a setting studied before in Sec.~III.C and IV.B of \citet{DFT-graphs}. There it was already noted that non-uv densities occur at the diagonals of the middle plane of the octahedron $\mathcal{P}_{4,2}$ that correspond to line-segments between potentials leading to degeneracy, but no picture of the full degeneracy structure was provided. We show this in Figure~\ref{fig:square-graph-deg-structure} that exhibits a pillow-like structure composed from families of ellipses (degeneracy bundles) that is symmetric around the middle plane of the octahedron $\mathcal{P}_{4,2}$. Those ellipses get flatter as the potential strength increases and they approach the edge of the density domain. For $v=0$, on the other hand, we have $\orho{1}=\orho{2}=\frac{1}{2}(1,1,1,1)$ and $\orho{12}=\frac{1}{2}(1,-1,1,-1)$ as the columns of $P$ which means $(g,\kappa) = (2,1)$ and thus, as it was described in Section~\ref{sec:density-regions}, the degeneracy region collapses into a vertical line segment in the center of $\mathcal{P}_{4,2}$.

\begin{figure*}[ht]
  \begin{subfigure}{.45\textwidth}
    \centering
    \includegraphics[width=\columnwidth]{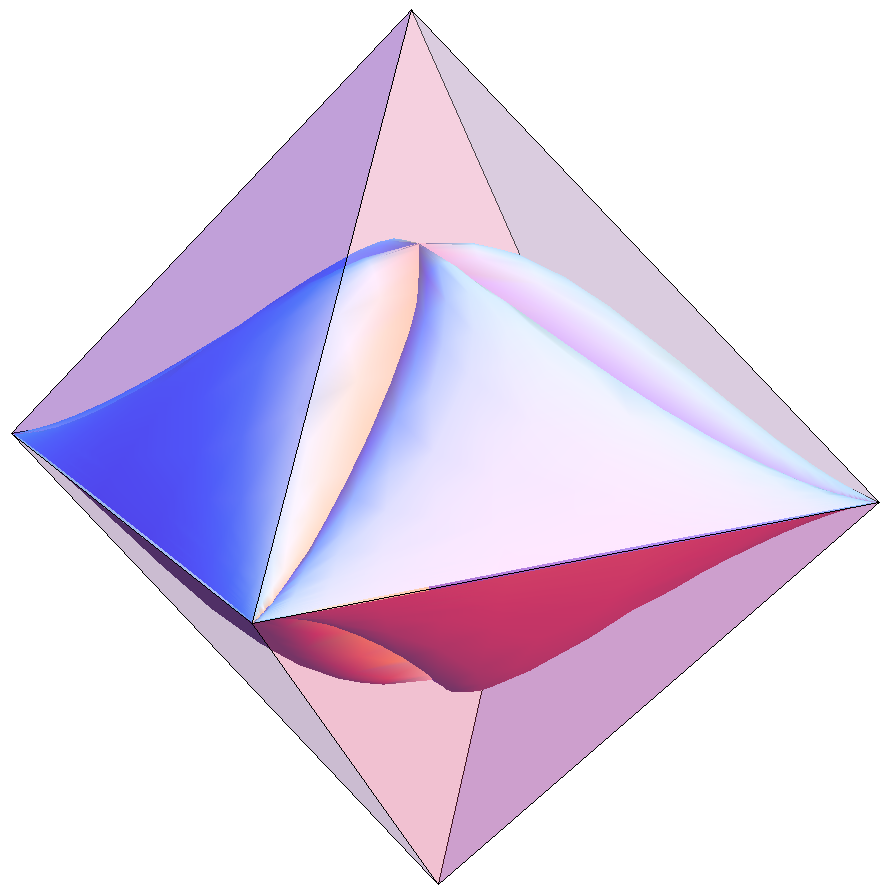}
    \caption{}
  \end{subfigure}
  \hfill 
  \begin{subfigure}{.45\textwidth}
    \centering
    \includegraphics[width=\columnwidth]{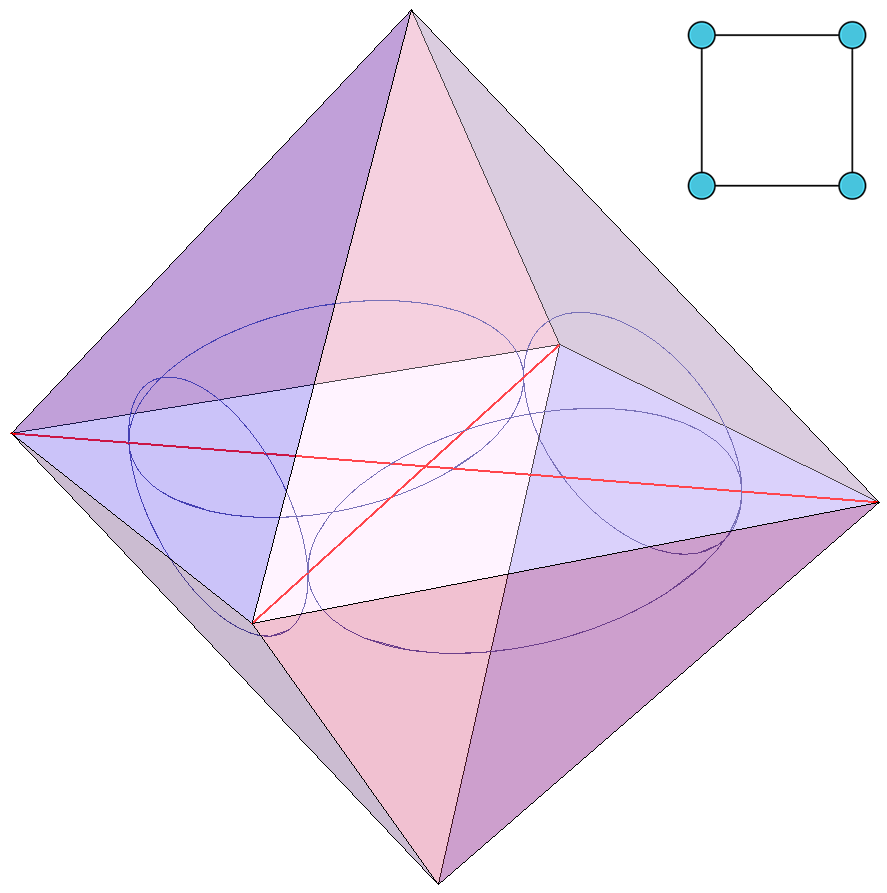}
    \caption{}
  \end{subfigure}
  \caption{Left, the union of all degeneracy regions of the square graph inside the octahedral set of densities $\mathcal{P}_{4,2}$. Right, four elliptical degeneracy regions from it are displayed that mutually touch at the diagonals drawn in red.}
  \label{fig:square-graph-deg-structure}
\end{figure*}

It is interesting to note that the geometrical features discussed here also appear in the universal functional $F(\rho)$. Since all densities $\rho\in D(v)=\overline\partial E(v)$ in a degeneracy region arise from a single potential and $v\in -\underline\partial F(\rho)$ as discussed in Section~\ref{sec:non-pure-v-rep}, this means that $F(\rho)$ is always flat on $D(v)$. Conversely, $F(\rho)$ is locally strictly convex outside of degeneracy regions. When two degeneracy regions touch at a non-uv $\rho$ then $F(\rho)$ displays a kink between the flat regions at that point which is responsible for non-differentiability of $F(\rho)$ and a subdifferential $-\underline\partial F(\rho)$ containing more than a single potential. Such a situation is displayed in Figure~\ref{fig:square-graph-F}, where the example of the square graph is repeated limited to densities on the middle plane of the octahedral density domain $\mathcal{P}_{4,2}$ with touching degeneracy regions along the diagonals.
\begin{figure}[ht]
    \centering
    \includegraphics[width=.9\columnwidth]{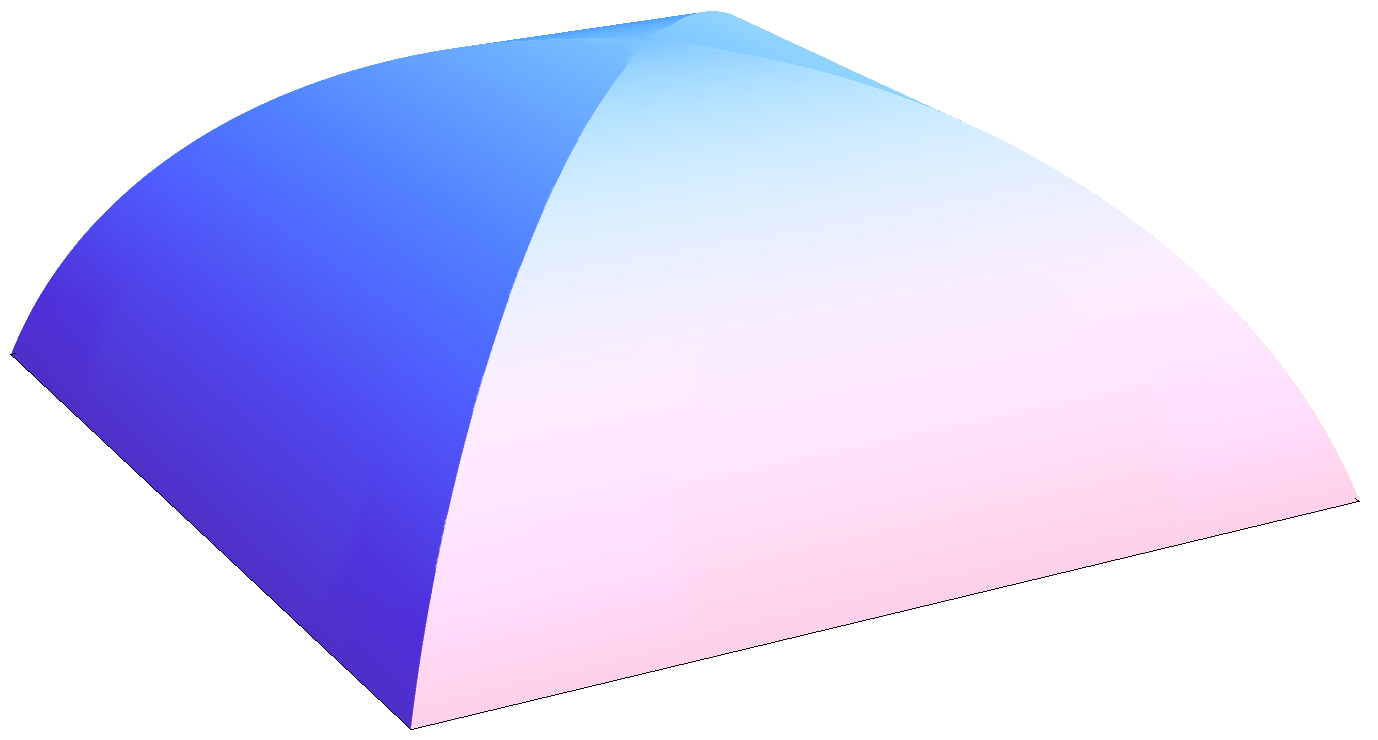}
    \caption{The universal functional $-F(\rho)$ displayed on the middle plane of the density octahedron from Figure~\ref{fig:square-graph-deg-structure}. Where degeneracy regions touch, the functional becomes non-differentiable, while it is constant along the lines that are inside an elliptical degeneracy region. Note that it is differentiable again on the center point.}
    \label{fig:square-graph-F}
\end{figure}

\section{Summary and open questions}
\label{sec:summary}

In summary, we analyzed the geometrical structures arising from degeneracy in the density domain and in potential space. While a state eigenspace is mapped to intricate shapes given by the convex hulls of algebraic varieties (Sections~\ref{sec:subspaces-density-map} and \ref{sec:density-regions}) and their geometric properties tell when a density is not pure-state $v$-representable (Section~\ref{sec:non-pure-v-rep}, those shapes can even form bundles of a certain dimensionality when the potential is varied (Section~\ref{sec:Ullrich-Kohn}). Finally, by just observing where those objects, the degeneracy regions, touch each other or the boundary of the density domain, one finds densities that are represented by multiple potentials (non-uv densities), a situation where the Hohenberg--Kohn theorem from density-functional theory fails. Overall, a beautiful geometrical picture of the relation between potentials and their ground-state densities emerges.

Looking at the different examples that have been discussed here reveals that a substantial part of the whole density domain is filled by the degeneracy structures. We might thus wonder about the precise \emph{degeneracy ratio}, i.e., the volume of the complete degeneracy structure divided by the volume of the hypersimplex $\mathcal{P}_{M,N}$. For the triangle graph and the square graph this can be calculated by elementary methods and we get $0.605$ and $0.589$, respectively. To calculate the volume of the convexified Roman surface is much harder, but we can give a lower and an upper estimate by considering a tetrahedron inside the Roman surface and one enclosing it. The inner tetrahedron shares the four points where the Roman surface touches the boundary of $\mathcal{P}_{3,2}$ as vertices. The outer tetrahedron is taken such that it shares the flat sides of the convexified Roman surface as faces, which means that the cylindrical degeneracy bundles lead exactly to the incircles of the tetrahedron faces. This gives a degeneracy ratio for the tetrahedron-graph system between $0.528$ and $0.703$. Hence, in all those examples that are highly symmetric and thus can give rise to degeneracy more easily, the degeneracy ratio is about 60\%, still an unexpectedly large number. Such a large amount of degeneracy was also observed in the discretization of continuum systems \cite{garrigue2021KS}, so it is not exclusively an effect in small lattice systems.

We finally list a number of open challenges and ideas for further investigations linked to the theory developed in this work:

\begin{itemize}
    \item One obvious open task would be to completely classify all possible shapes of degeneracy regions for general $(g,\kappa)$ following the analysis of \citet{degen1994} and \citet{Coffman-Roman-surf} that settle the case $(g,\kappa)=(3,2)$. Some could already be ruled out by Lemma~\ref{lem:VKerPzero}, but possibly other classes also never appear as degeneracy regions, which could be demonstrated by further investigations related to the nature of $\ker P$. After all, in our examples of all Steiner surfaces only the Roman surface appeared. The case of a singleton degeneracy region with $\kappa=(g+1)g/2-1$ was also never observed in an example and it is open if it can even occur without internal degrees of freedom. If it can be proven that they do not exist then this would save the original argument of \citet{ullrich2002} that degeneracy in potential space is rare.
    
    \item In order to find more classes of degeneracy regions or whole degeneracy structures, one can naturally study other interesting finite-lattice system. For the preparation of this work graph systems apart from the triangle, square, tetrahedron and cuboctahedron with two particles have been studied: the diamond graph shows only accidental degeneracy for a certain symmetrical potential with $(g,\kappa)=(2,1)$, the claw graph just one cylindrical degeneracy bundle, the octahedron graph has a degeneracy region with $(g,\kappa)=(3,3)$ that has the shape of a triangle, and the cube graph leads to bundles of $(g,\kappa)=(3,0)$ regions (this yields another example for non-pure-state $v$-representability by Corollary~\ref{cor:non-pure-state-v-rep}), further cylindrical bundles and a Roman surface in the center. But this list is far from exhaustive.
    
    \item It was already noted in Section~\ref{sec:geom} that cases of whole degeneracy regions as the intersection between two larger degeneracy regions or at the boundary have not yet been observed. The question remains if they are even possible. It also remains open if touching the boundary of $\mathcal{P}_{M,N}$ only occurs at the highest-dimensional faces of the hypersimplex, where all $\rho_i>0$, or also at lower-dimensional edges. In case of the former and since by Theorem~\ref{th:geometry} densities on the border are only $v$-representable if they are contained in a degeneracy region, this means that densities with $\rho_i=0$ can never occur for any potential $v$ in a finite-lattice system or, in other words, all ground-state densities are strictly positive. Some corresponding results for non-interacting or weakly-interacting lattice systems can be found in \citet[Sec.~V.A]{DFT-graphs}.
    
    \item Non-uv densities and especially densities from degeneracy regions, since they are more frequent, can pose a problem to the convergence of the Kohn--Sham iteration by introducing ambiguities \cite{garrigue2021KS}. Imagine the target density of an interacting system that lies within a degeneracy region of the non-interacting reference system (this can indeed happen, since for example a Hamiltonian with a central potential commutes with the Laplace--Runge--Lenz vector, leading to additional symmetries, while the Coulomb interaction breaks this symmetry), then even if we find an almost correct effective potential we would lie outside of the degeneracy region and thus potentially stay far away from the target density.
    
    \item In the end of Section~\ref{sec:geom} we noted that $F(\rho)$ is flat on degeneracy regions, a feature that will not be respected by the usual approximations to the universal functional. More insight into its shape, arising from the discussed geometrical structure, can thus improve such approximations.
    
    \item Moreau--Yosida regularization is applied to the universal functional $F(\rho)$ in order to remedy non-differentiability of the functional \cite{Kvaal2014,laestadius2018generalized} and to help proving convergence of the Kohn--Sham iteration in the finite lattice case \cite{penz2019guaranteed,penz2020erratum}. In the geometrical picture, this leads to degeneracy regions that shrink slightly, thus moving apart from each other, and a boundary of the density domain that is entirely removed. This then avoids any non-uv densities by Theorem~\ref{th:geometry}.
    
    \item Further dual settings apart from ``density-potential'' with even more complicated geometries can be considered as well. For example current DFT maps the tuple $(\rho,\mathbf{j})$ including the current density to $(v,\mathbf{A})$, where the potential is complemented by a vector potential \cite{laestadius2019CDFT}. Since Theorem~\ref{th:HK1} also holds for this case, a corresponding geometrical analysis can be conducted.
It is known that non-uv densities do also occur in current DFT in the continuum case, where they were shown to arise in situations where one also has ground-state degeneracy \cite{LaestadiusTellgren2018}.
    
    \item Theorem~\ref{th:geometry} exhibits an interesting dual structure between the geometric elements in the density and potential spaces introduced here. That the spaces themselves are topological duals of each other was already noted, but we also see that each degeneracy region (as an extended object) is connected to a single density point, while the straight lines connecting those points correspond to the density where the degeneracy regions touch. This structure is very similar to the \emph{nerve complex} of a set family.
    
\end{itemize}

\begin{acknowledgements}
The foundation of this work was laid at a joint research visit in November 2021 at the Oberwolfach Research Institute for Mathematics as a part of their ``Research in Pairs'' program. Hence, our special gratitude goes to the institute for inviting us, to Michael Ruggenthaler who joined us there, and to Kevin Piterman who was a visiting fellow at that time and helped us with some questions in algebraic geometry. 
Coincidentally, the Roman surface was only recognized at first because of a photo showing a sculpture of it in the immense library at Oberwolfach.
R.~v.~L.\ further acknowledges the Academy of Finland for support under project no.~317139.
\end{acknowledgements}

\begin{appendix}
\section{Extension to complex hermitian Hamiltonians}

The general assumption throughout the work above was a real symmetric Hamiltonian $H_v$ so that one can choose \emph{real} basis vectors $\{\Phi_k\}_{k=1}^g$  for any eigenspace. This in turn was important for being able to define \emph{real} matrix elements $\orho{kl} = 2\langle \Phi_k,\hat\rho\Phi_l\rangle$ that enter the projection map $P$ that maps the Veronese variety $\mathbb{V}_g$ onto $\Dens$ (see Section~\ref{sec:subspaces-density-map}). We will here extend this investigation to general complex hermitian Hamiltonians.
The definitions for the density regions of a given subspace $\U$ (\emph{degeneracy regions} if it is an eigenspace of $H_v$),
\begin{subequations}
\begin{align}
    &D_\C = \rho(\U), \\
    &D = \rho(\{ \Gamma \in \mathcal{D} \mid \Gamma\H \subseteq \U \}),
\end{align}
\end{subequations}
stay entirely the same. Only the $D_\R = \rho(\mathcal{U}_\R)$ makes no sense any more for a subspace with complex basis vectors, since it would be basis dependent. But we still have the relation $D=\ch D_\C$ which follows directly from Eq.~\eqref{eq:rho-from-Gamma}.

Let $z=(z_1,\ldots,z_g) \in\C^g$, $\|z\|=1$, be the coordinates for a normalized state $\Psi\in\mathcal{U}$ with respect to the (complex) orthonormal subspace basis $\{\Phi_k\}_{k=1}^g$ of $\mathcal{U}$. Then the density is evaluated as
\begin{equation}\label{eq:rho-from-x-complex}
    \rho(\Psi) = \sum_{k=1}^g |z_k|^2 \rho(\Phi_k) + \sum_{\substack{k,l=1\\k< l}}^g 2\Re ( z_k^*z_l \langle \Phi_k,\hat\rho\Phi_l\rangle ).
\end{equation}
Note that the difference to Eq.~\eqref{eq:rho-from-x} is just that the real-part function is applied also to the inner-product expression. In analogy to before we define
\begin{subequations}
\begin{align}
	&\orho{k} = \rho(\Phi_k), \\
	&\orhoRe{kl} = 2\Re ( \langle \Phi_k,\hat\rho\Phi_l\rangle ), \\
	&\orhoIm{kl} = 2\Im ( \langle \Phi_k,\hat\rho\Phi_l\rangle ).
\end{align}
\end{subequations}
With this we rewrite Eq.~\eqref{eq:rho-from-x-complex} as
\begin{equation}
    \rho(\Psi) = \sum_{k=1}^g |z_k|^2 \orho{k} + \sum_{\substack{k,l=1\\k< l}}^g \left( \Re ( z_k^*z_l ) \orhoRe{kl} - \Im ( z_k^*z_l ) \orhoIm{kl} \right).
\end{equation}
This form suggests a generalized Veronese embedding
\begin{equation}\begin{aligned}
    \nu_\C : \C^g \quad\quad\quad\;\; &\longrightarrow \R^{g^2} \\
    (z_1,\ldots,z_g) &\longmapsto (|z_1|^2,\ldots,|z_g|^2,\\
    &\quad\quad\;\; \Re(z_1^*z_2),\ldots,\Re(z_{g-1}^*z_g),\\
    &\quad\quad\;\; \Im(z_1^*z_2),\ldots,\Im(z_{g-1}^*z_g)).
\end{aligned}\end{equation}
Finally, the map from coordinates $z\in\C^g$ to densities is composed of this $\nu_\C$ and a linear map $P: \R^{g^2} \to \Dens$ with columns $\orho{k}$, $\orhoRe{kl}$, and $\orhoIm{kl}$, just like Eq.~\eqref{eq:P-nu-map},
\begin{equation}
	\rho = P \circ \nu_\C : \C^g \mapsto \Dens.
\end{equation}
Whenever such a column is zero or linearly dependent on the rest, the dimension of the kernel $\kappa = \dim\ker P$ increases and the dimension of the formed density region decreases. So the new dimensional formula for a density region is then
\begin{equation}
	\dim D = g^2 - \kappa - 1.
\end{equation}
Theorem~\ref{th:Ullrich-Kohn} still holds unchanged, with $(g+1)g/2$ replaced by $g^2$, as well as the whole Section~\ref{sec:geom}.

Using a complex Hamiltonian allows for an even simpler realization of the Roman surface as a degeneracy region. For this we take $M=4$ and $N=1$, so just one particle and a 4-dimensional Hilbert space, with the Hamiltonian
\begin{equation}
	h_0 = \begin{pmatrix*}[r]
		0 & \i & \i & \i\\
		-\i & 0 & 1 & 1\\
		-\i & 1 & 0 & 1\\
		-\i & 1 & 1 & 0
	\end{pmatrix*}.
\end{equation}
Unlike the Hamiltonian considered in Section~\ref{sec:density-regions}, the Perron--Frobenius theorem does not apply complex Hamiltonians and we find that the ground state indeed is 3-fold degenerate with an eigenspace spanned by
\begin{align}
	&\Phi_1 = \frac{1}{\sqrt{2}} (\i,0,0,-1),\\
	&\Phi_2 = \frac{1}{\sqrt{6}} (\i,0,-2,1),\\
	&\Phi_3 = \frac{1}{2\sqrt{3}} (\i,-3,1,1).
\end{align}
The degeneracy region from this eigenspace is now again the convex hull of a Roman surface that lies in the density domain $\mathcal{P}_{4,1}$ which has the shape of a tetrahedron. It is displayed in Fig.~\ref{fig:Roman-surf-Tetrahedron}.
\begin{figure}[ht]
    \centering
    \includegraphics[width=.9\columnwidth]{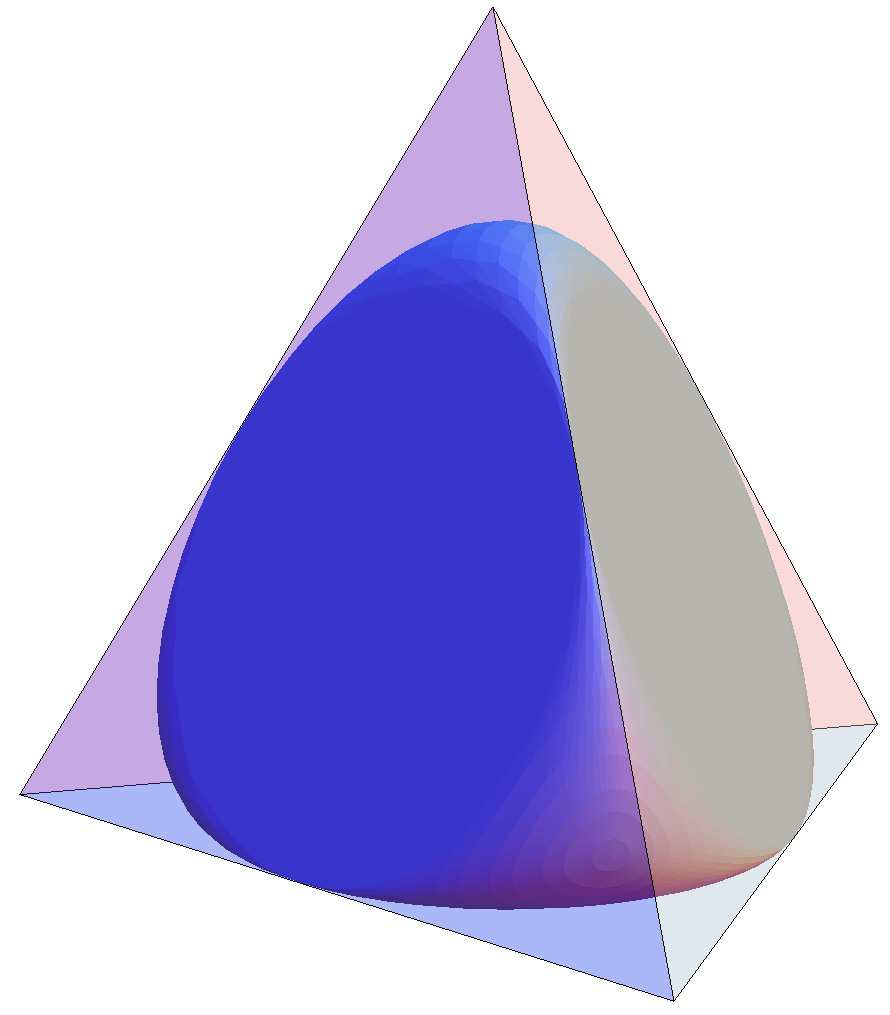}
    \caption{The convex, three-fold degeneracy region $D$ inside the tetrahedronal density domain. The corners of the octahedron correspond to the extreme density $(1,0,0,0)$ and its permutations, generalized barycentric coordinates are used to display densities.}
    \label{fig:Roman-surf-Tetrahedron}
\end{figure}
We note three things on how the degeneracy region relates to the density domain: (i) it almost fills out the whole domain, so the degeneracy ratio is high, (ii) the flat faces of the convex hull of the Roman surface touch the faces of the tetrahedron, and (iii) the degeneracy region also touches every edge of the tetrahedron in exactly one point. The last item gives a counterexample to a question raised in Section~\ref{sec:summary} earlier: if touching the boundary of the density domain only occurs at the highest-dimensional faces. This implies that by Theorem~\ref{th:geometry} this examples features exactly six $v$-representable densities (one in the middle of every edge) that have even two components of $\rho$ equal to zero (the $\Phi_1$ above gives one such example, i.e., $\rho_1=\frac{1}{2}(1,0,0,1)$). Additionally, since also the flat faces touch the density domain, all those densities are $v$-representable as well with one component of $\rho$ equal to zero (for this, the $\Phi_2$ above gives an example, i.e., $\rho_2=\frac{1}{6}(1,0,4,1)$). These are thus all examples of non-uv densities. We conclude that as far as \emph{complex} Hamiltonians are concerned, the ground-state density must not be strictly positive.

\section{Extension to sets of arbitrary observables}

Another form of extension is almost obvious. For an arbitrary quantum system, not limited to finite lattices, take a set of observables $\{\hat o_1,\ldots,\hat o_M\}$, all represented by self-adjoint operators. For any normalized state $\Psi\in\H$ we define
\begin{equation}
	o(\Psi,\Phi) = (\langle\Psi,\hat o_1\Phi\rangle, \ldots, \langle\Psi,\hat o_M\Phi\rangle).
\end{equation}
Then, replacing $\rho(\Psi)$ by $o(\Psi,\Psi)$ and extending this definition to mixed states like in Eq.~\eqref{eq:rho-from-Gamma}, we can have equivalent definitions for density (or degeneracy) regions. The linear map $P:\R^{g^2}\to\R^M$ now consists of $o(\Phi_k,\Phi_k)$, $2\Re(o(\Phi_k,\Phi_l))$, and $2\Im(o(\Phi_k,\Phi_l))$, and the whole dimensional and morphological analyses of Sections~\ref{sec:subspaces-density-map}-\ref{sec:non-pure-v-rep} still hold.
Just note that all the dimensional formulas, starting with Eq.~\eqref{eq:max-dim-D}, change because of the missing normalization constraint.
Also, the result $D(v) = \overline\partial E(v)$ in Section~\ref{sec:non-pure-v-rep} depends on the dual relation to potentials discussed below.
The only possible issue occurs with unbounded operators for observables that could lead to infinities. Conversely, the results of Sections~\ref{sec:Ullrich-Kohn}-\ref{sec:geom} do not readily apply, since they depend intimately on the dual relation between densities and potentials originating from the linear coupling $\sum_i v_i \hat \rho_i$. But if the Hamiltonian includes a coupling of the form $\sum_i v_i \hat o_i$ with a generalized potential $v$, this would again yield the same structure and then everything, including $D(v) = \overline\partial E(v)$, can be derived in the same fashion.

The same construction and most notably the image of the map $\Psi\mapsto o(\Psi,\Psi)$, called the `moduli space of expectation values', was recently studied by \citet{Song2023}. In case the mapping is limited to a subspace this moduli space amounts exactly to our definition of density region. The reference discusses interesting applications to estimates on expectation valued, like Heisenberg's uncertainty principle and Bell's inequality, and even a direct application to DFT.

A typical example for such a Hamiltonian would be a lattice Hamiltonian including spin coupled to a magnetic field $\vec B_i$,
\begin{equation}
	H_B = H_0 + \sum_i \vec B_i \cdot \hat{\vec{\sigma}}_i.
\end{equation}
Another interesting example is one-particle reduced density matrix functional theory (1RDMFT)~\cite{Liebert2023-1RDMFT}, where instead of the density one considers the 1RDM and the coupling is to a non-local potential,
\begin{equation}
	H_v = H_0 + \sum_{i,j} v_{ij} \hat{\gamma}_{ij}.
\end{equation}

\end{appendix}

\newpage


\begin{thebibliography}{48}%
\makeatletter
\providecommand \@ifxundefined [1]{%
 \@ifx{#1\undefined}
}%
\providecommand \@ifnum [1]{%
 \ifnum #1\expandafter \@firstoftwo
 \else \expandafter \@secondoftwo
 \fi
}%
\providecommand \@ifx [1]{%
 \ifx #1\expandafter \@firstoftwo
 \else \expandafter \@secondoftwo
 \fi
}%
\providecommand \natexlab [1]{#1}%
\providecommand \enquote  [1]{``#1''}%
\providecommand \bibnamefont  [1]{#1}%
\providecommand \bibfnamefont [1]{#1}%
\providecommand \citenamefont [1]{#1}%
\providecommand \href@noop [0]{\@secondoftwo}%
\providecommand \href [0]{\begingroup \@sanitize@url \@href}%
\providecommand \@href[1]{\@@startlink{#1}\@@href}%
\providecommand \@@href[1]{\endgroup#1\@@endlink}%
\providecommand \@sanitize@url [0]{\catcode `\\12\catcode `\$12\catcode
  `\&12\catcode `\#12\catcode `\^12\catcode `\_12\catcode `\%12\relax}%
\providecommand \@@startlink[1]{}%
\providecommand \@@endlink[0]{}%
\providecommand \url  [0]{\begingroup\@sanitize@url \@url }%
\providecommand \@url [1]{\endgroup\@href {#1}{\urlprefix }}%
\providecommand \urlprefix  [0]{URL }%
\providecommand \Eprint [0]{\href }%
\providecommand \doibase [0]{http://dx.doi.org/}%
\providecommand \selectlanguage [0]{\@gobble}%
\providecommand \bibinfo  [0]{\@secondoftwo}%
\providecommand \bibfield  [0]{\@secondoftwo}%
\providecommand \translation [1]{[#1]}%
\providecommand \BibitemOpen [0]{}%
\providecommand \bibitemStop [0]{}%
\providecommand \bibitemNoStop [0]{.\EOS\space}%
\providecommand \EOS [0]{\spacefactor3000\relax}%
\providecommand \BibitemShut  [1]{\csname bibitem#1\endcsname}%
\let\auto@bib@innerbib\@empty
\bibitem [{\citenamefont {von Barth}(2004)}]{vonBarth2004basic}%
  \BibitemOpen
  \bibfield  {author} {\bibinfo {author} {\bibfnamefont {U.}~\bibnamefont {von
  Barth}},\ }\bibfield  {title} {\emph {\bibinfo {title} {Basic
  density-functional theory—an overview},\ }}\href {\doibase
  10.1238/Physica.Topical.109a00009} {\bibfield  {journal} {\bibinfo  {journal}
  {Phys. Scr.}\ }\textbf {\bibinfo {volume} {2004}},\ \bibinfo {pages} {9}
  (\bibinfo {year} {2004})}\BibitemShut {NoStop}%
\bibitem [{\citenamefont {Burke}\ and\ \citenamefont
  {friends}(2007)}]{burke2007abc}%
  \BibitemOpen
  \bibfield  {author} {\bibinfo {author} {\bibfnamefont {K.}~\bibnamefont
  {Burke}}\ and\ \bibinfo {author} {\bibnamefont {friends}},\ }\href
  {https://dft.uci.edu/doc/g1.pdf} {\bibinfo {title} {The {ABC} of {DFT}},\ }
  (\bibinfo {year} {2007})\BibitemShut {NoStop}%
\bibitem [{\citenamefont {Dreizler}\ and\ \citenamefont
  {Gross}(2012)}]{dreizler2012-book}%
  \BibitemOpen
  \bibfield  {author} {\bibinfo {author} {\bibfnamefont {R.~M.}\ \bibnamefont
  {Dreizler}}\ and\ \bibinfo {author} {\bibfnamefont {E.~K.}\ \bibnamefont
  {Gross}},\ }\href@noop {} {\emph {\bibinfo {title} {Density functional
  theory: An approach to the quantum many-body problem}}}\ (\bibinfo
  {publisher} {Springer},\ \bibinfo {year} {2012})\BibitemShut {NoStop}%
\bibitem [{\citenamefont {Eschrig}(2003)}]{eschrig2003-book}%
  \BibitemOpen
  \bibfield  {author} {\bibinfo {author} {\bibfnamefont {H.}~\bibnamefont
  {Eschrig}},\ }\href@noop {} {\emph {\bibinfo {title} {The fundamentals of
  density functional theory}}},\ \bibinfo {edition} {2nd}\ ed.\ (\bibinfo
  {publisher} {Springer},\ \bibinfo {year} {2003})\BibitemShut {NoStop}%
\bibitem [{\citenamefont {Ullrich}(2011)}]{ullrich2011-book}%
  \BibitemOpen
  \bibfield  {author} {\bibinfo {author} {\bibfnamefont {C.~A.}\ \bibnamefont
  {Ullrich}},\ }\href@noop {} {\emph {\bibinfo {title} {Time-dependent
  density-functional theory: {C}oncepts and applications}}}\ (\bibinfo
  {publisher} {OUP Oxford},\ \bibinfo {year} {2011})\BibitemShut {NoStop}%
\bibitem [{\citenamefont {Ullrich}\ and\ \citenamefont
  {Yang}(2014)}]{ullrich2014brief}%
  \BibitemOpen
  \bibfield  {author} {\bibinfo {author} {\bibfnamefont {C.~A.}\ \bibnamefont
  {Ullrich}}\ and\ \bibinfo {author} {\bibfnamefont {Z.}~\bibnamefont {Yang}},\
  }\bibfield  {title} {\emph {\bibinfo {title} {A brief compendium of
  time-dependent density functional theory},\ }}\href {\doibase
  10.1007/s13538-013-0141-2} {\bibfield  {journal} {\bibinfo  {journal} {Braz.
  J. Phys.}\ }\textbf {\bibinfo {volume} {44}},\ \bibinfo {pages} {154}
  (\bibinfo {year} {2014})}\BibitemShut {NoStop}%
\bibitem [{\citenamefont {Vignale}\ and\ \citenamefont
  {Rasolt}(1987)}]{Vignale1987}%
  \BibitemOpen
  \bibfield  {author} {\bibinfo {author} {\bibfnamefont {G.}~\bibnamefont
  {Vignale}}\ and\ \bibinfo {author} {\bibfnamefont {M.}~\bibnamefont
  {Rasolt}},\ }\bibfield  {title} {\emph {\bibinfo {title} {Density-functional
  theory in strong magnetic fields},\ }}\href {\doibase
  10.1103/PhysRevLett.59.2360} {\bibfield  {journal} {\bibinfo  {journal}
  {Phys. Rev. Lett.}\ }\textbf {\bibinfo {volume} {59}},\ \bibinfo {pages}
  {2360} (\bibinfo {year} {1987})}\BibitemShut {NoStop}%
\bibitem [{\citenamefont {Vignale}(2004)}]{VIGNALE_PRB70_201102}%
  \BibitemOpen
  \bibfield  {author} {\bibinfo {author} {\bibfnamefont {G.}~\bibnamefont
  {Vignale}},\ }\bibfield  {title} {\emph {\bibinfo {title} {Mapping from
  current densities to vector potentials in time-dependent current density
  functional theory},\ }}\href {\doibase 10.1103/PhysRevB.70.201102} {\bibfield
   {journal} {\bibinfo  {journal} {Phys. Rev. B}\ }\textbf {\bibinfo {volume}
  {70}},\ \bibinfo {pages} {201102} (\bibinfo {year} {2004})}\BibitemShut
  {NoStop}%
\bibitem [{\citenamefont {Ruggenthaler}\ \emph {et~al.}(2014)\citenamefont
  {Ruggenthaler}, \citenamefont {Flick}, \citenamefont {Pellegrini},
  \citenamefont {Appel}, \citenamefont {Tokatly},\ and\ \citenamefont
  {Rubio}}]{ruggenthaler2014-QEDFT}%
  \BibitemOpen
  \bibfield  {author} {\bibinfo {author} {\bibfnamefont {M.}~\bibnamefont
  {Ruggenthaler}}, \bibinfo {author} {\bibfnamefont {J.}~\bibnamefont {Flick}},
  \bibinfo {author} {\bibfnamefont {C.}~\bibnamefont {Pellegrini}}, \bibinfo
  {author} {\bibfnamefont {H.}~\bibnamefont {Appel}}, \bibinfo {author}
  {\bibfnamefont {I.~V.}\ \bibnamefont {Tokatly}}, \ and\ \bibinfo {author}
  {\bibfnamefont {A.}~\bibnamefont {Rubio}},\ }\bibfield  {title} {\emph
  {\bibinfo {title} {Quantum-electrodynamical density-functional theory:
  {B}ridging quantum optics and electronic-structure theory},\ }}\href
  {\doibase 10.1103/PhysRevA.90.012508} {\bibfield  {journal} {\bibinfo
  {journal} {Phys. Rev. A}\ }\textbf {\bibinfo {volume} {90}},\ \bibinfo
  {pages} {012508} (\bibinfo {year} {2014})}\BibitemShut {NoStop}%
\bibitem [{\citenamefont {Ullrich}\ and\ \citenamefont
  {Kohn}(2002)}]{ullrich2002}%
  \BibitemOpen
  \bibfield  {author} {\bibinfo {author} {\bibfnamefont {C.~A.}\ \bibnamefont
  {Ullrich}}\ and\ \bibinfo {author} {\bibfnamefont {W.}~\bibnamefont {Kohn}},\
  }\bibfield  {title} {\emph {\bibinfo {title} {Degeneracy in density
  functional theory: {T}opology in the v and n spaces},\ }}\href {\doibase
  10.1103/PhysRevLett.89.156401} {\bibfield  {journal} {\bibinfo  {journal}
  {Phys. Rev. Lett.}\ }\textbf {\bibinfo {volume} {89}},\ \bibinfo {pages}
  {156401} (\bibinfo {year} {2002})}\BibitemShut {NoStop}%
\bibitem [{\citenamefont {Garrigue}(2021)}]{garrigue2021-potential-density}%
  \BibitemOpen
  \bibfield  {author} {\bibinfo {author} {\bibfnamefont {L.}~\bibnamefont
  {Garrigue}},\ }\bibfield  {title} {\emph {\bibinfo {title} {Some properties
  of the potential-to-ground state map in quantum mechanics},\ }}\href
  {\doibase 10.1007/s00220-021-04140-9.pdf} {\bibfield  {journal} {\bibinfo
  {journal} {Commun. Math. Phys.}\ }\textbf {\bibinfo {volume} {386}},\
  \bibinfo {pages} {1803} (\bibinfo {year} {2021})}\BibitemShut {NoStop}%
\bibitem [{\citenamefont {Arovas}\ \emph {et~al.}(2022)\citenamefont {Arovas},
  \citenamefont {Berg}, \citenamefont {Kivelson},\ and\ \citenamefont
  {Raghu}}]{arovas2022hubbard}%
  \BibitemOpen
  \bibfield  {author} {\bibinfo {author} {\bibfnamefont {D.~P.}\ \bibnamefont
  {Arovas}}, \bibinfo {author} {\bibfnamefont {E.}~\bibnamefont {Berg}},
  \bibinfo {author} {\bibfnamefont {S.~A.}\ \bibnamefont {Kivelson}}, \ and\
  \bibinfo {author} {\bibfnamefont {S.}~\bibnamefont {Raghu}},\ }\bibfield
  {title} {\emph {\bibinfo {title} {The {H}ubbard model},\ }}\href {\doibase
  10.1146/annurev-conmatphys-031620-102024} {\bibfield  {journal} {\bibinfo
  {journal} {Annu. Rev. Condens. Matter Phys.}\ }\textbf {\bibinfo {volume}
  {13}},\ \bibinfo {pages} {239} (\bibinfo {year} {2022})}\BibitemShut
  {NoStop}%
\bibitem [{\citenamefont {Qin}\ \emph {et~al.}(2022)\citenamefont {Qin},
  \citenamefont {Sch{\"a}fer}, \citenamefont {Andergassen}, \citenamefont
  {Corboz},\ and\ \citenamefont {Gull}}]{qin2022hubbard}%
  \BibitemOpen
  \bibfield  {author} {\bibinfo {author} {\bibfnamefont {M.}~\bibnamefont
  {Qin}}, \bibinfo {author} {\bibfnamefont {T.}~\bibnamefont {Sch{\"a}fer}},
  \bibinfo {author} {\bibfnamefont {S.}~\bibnamefont {Andergassen}}, \bibinfo
  {author} {\bibfnamefont {P.}~\bibnamefont {Corboz}}, \ and\ \bibinfo {author}
  {\bibfnamefont {E.}~\bibnamefont {Gull}},\ }\bibfield  {title} {\emph
  {\bibinfo {title} {The {H}ubbard model: A computational perspective},\
  }}\href {\doibase 10.1146/annurev-conmatphys-090921-033948} {\bibfield
  {journal} {\bibinfo  {journal} {Annu. Rev. Condens. Matter Phys.}\ }\textbf
  {\bibinfo {volume} {13}},\ \bibinfo {pages} {275} (\bibinfo {year}
  {2022})}\BibitemShut {NoStop}%
\bibitem [{\citenamefont {Flores}\ \emph {et~al.}(2022)\citenamefont {Flores},
  \citenamefont {Soler-Polo},\ and\ \citenamefont
  {Ortega}}]{flores2022localorbital}%
  \BibitemOpen
  \bibfield  {author} {\bibinfo {author} {\bibfnamefont {F.}~\bibnamefont
  {Flores}}, \bibinfo {author} {\bibfnamefont {D.}~\bibnamefont {Soler-Polo}},
  \ and\ \bibinfo {author} {\bibfnamefont {J.}~\bibnamefont {Ortega}},\
  }\bibfield  {title} {\emph {\bibinfo {title} {A closed local-orbital unified
  description of dft and many-body effects},\ }}\href {\doibase
  10.1088/1361-648X/ac6eae} {\bibfield  {journal} {\bibinfo  {journal} {J.
  Phys. Condens. Matter}\ }\textbf {\bibinfo {volume} {34}},\ \bibinfo {pages}
  {304006} (\bibinfo {year} {2022})}\BibitemShut {NoStop}%
\bibitem [{\citenamefont {Penz}\ and\ \citenamefont {van
  Leeuwen}(2021)}]{DFT-graphs}%
  \BibitemOpen
  \bibfield  {author} {\bibinfo {author} {\bibfnamefont {M.}~\bibnamefont
  {Penz}}\ and\ \bibinfo {author} {\bibfnamefont {R.}~\bibnamefont {van
  Leeuwen}},\ }\bibfield  {title} {\emph {\bibinfo {title} {Density-functional
  theory on graphs},\ }}\href {\doibase 10.1063/5.0074249} {\bibfield
  {journal} {\bibinfo  {journal} {J. Chem. Phys.}\ }\textbf {\bibinfo {volume}
  {155}},\ \bibinfo {pages} {244111} (\bibinfo {year} {2021})}\BibitemShut
  {NoStop}%
\bibitem [{\citenamefont {Lieb}(1983)}]{Lieb1983}%
  \BibitemOpen
  \bibfield  {author} {\bibinfo {author} {\bibfnamefont {E.~H.}\ \bibnamefont
  {Lieb}},\ }\bibfield  {title} {\emph {\bibinfo {title} {Density functionals
  for {C}oulomb-systems},\ }}\href {\doibase 10.1002/qua.560240302} {\bibfield
  {journal} {\bibinfo  {journal} {Int. J. Quantum Chem.}\ }\textbf {\bibinfo
  {volume} {24}},\ \bibinfo {pages} {243} (\bibinfo {year} {1983})}\BibitemShut
  {NoStop}%
\bibitem [{\citenamefont {Tellgren}\ \emph {et~al.}(2018)\citenamefont
  {Tellgren}, \citenamefont {Laestadius}, \citenamefont {Helgaker},
  \citenamefont {Kvaal},\ and\ \citenamefont {Teale}}]{Tellgren2018}%
  \BibitemOpen
  \bibfield  {author} {\bibinfo {author} {\bibfnamefont {E.~I.}\ \bibnamefont
  {Tellgren}}, \bibinfo {author} {\bibfnamefont {A.}~\bibnamefont
  {Laestadius}}, \bibinfo {author} {\bibfnamefont {T.}~\bibnamefont
  {Helgaker}}, \bibinfo {author} {\bibfnamefont {S.}~\bibnamefont {Kvaal}}, \
  and\ \bibinfo {author} {\bibfnamefont {A.~M.}\ \bibnamefont {Teale}},\
  }\bibfield  {title} {\emph {\bibinfo {title} {Uniform magnetic fields in
  density-functional theory},\ }}\href {\doibase 10.1063/1.5007300} {\bibfield
  {journal} {\bibinfo  {journal} {J. Chem. Phys.}\ }\textbf {\bibinfo {volume}
  {148}},\ \bibinfo {pages} {024101} (\bibinfo {year} {2018})}\BibitemShut
  {NoStop}%
\bibitem [{\citenamefont {Penz}\ \emph {et~al.}(2023)\citenamefont {Penz},
  \citenamefont {Tellgren}, \citenamefont {Csirik}, \citenamefont
  {Ruggenthaler},\ and\ \citenamefont {Laestadius}}]{dens-pot-review}%
  \BibitemOpen
  \bibfield  {author} {\bibinfo {author} {\bibfnamefont {M.}~\bibnamefont
  {Penz}}, \bibinfo {author} {\bibfnamefont {E.~I.}\ \bibnamefont {Tellgren}},
  \bibinfo {author} {\bibfnamefont {M.~A.}\ \bibnamefont {Csirik}}, \bibinfo
  {author} {\bibfnamefont {M.}~\bibnamefont {Ruggenthaler}}, \ and\ \bibinfo
  {author} {\bibfnamefont {A.}~\bibnamefont {Laestadius}},\ }\bibfield  {title}
  {\emph {\bibinfo {title} {The structure of the density-potential mapping.
  {P}art {I}: Standard density-functional theory},\ }}\href {\doibase
  10.1021/acsphyschemau.2c00069} {\bibfield  {journal} {\bibinfo  {journal}
  {ACS Phys. Chem. Au}\ }\textbf {\bibinfo {volume} {3}},\ \bibinfo {pages}
  {334–347} (\bibinfo {year} {2023})}\BibitemShut {NoStop}%
\bibitem [{\citenamefont {Lewin}\ \emph {et~al.}(2019)\citenamefont {Lewin},
  \citenamefont {Lieb},\ and\ \citenamefont {Seiringer}}]{LewinFunctionals}%
  \BibitemOpen
  \bibfield  {author} {\bibinfo {author} {\bibfnamefont {M.}~\bibnamefont
  {Lewin}}, \bibinfo {author} {\bibfnamefont {E.~H.}\ \bibnamefont {Lieb}}, \
  and\ \bibinfo {author} {\bibfnamefont {R.}~\bibnamefont {Seiringer}},\
  }\bibfield  {title} {\emph {\bibinfo {title} {Universal functionals in
  density functional theory},\ }}\href {\doibase 10.48550/arXiv.1912.10424}
  {\bibfield  {journal} {\bibinfo  {journal} {arXiv preprint}\ } (\bibinfo
  {year} {2019}),\ 10.48550/arXiv.1912.10424}\BibitemShut {NoStop}%
\bibitem [{\citenamefont {Garrigue}(2018)}]{Garrigue2018}%
  \BibitemOpen
  \bibfield  {author} {\bibinfo {author} {\bibfnamefont {L.}~\bibnamefont
  {Garrigue}},\ }\bibfield  {title} {\emph {\bibinfo {title} {Unique
  continuation for many-body {S}chr{\"o}dinger operators and the
  {H}ohenberg--{K}ohn theorem},\ }}\href {\doibase 10.1007/s11040-018-9287-z}
  {\bibfield  {journal} {\bibinfo  {journal} {Math. Phys. Anal. Geom.}\
  }\textbf {\bibinfo {volume} {21}},\ \bibinfo {pages} {27} (\bibinfo {year}
  {2018})}\BibitemShut {NoStop}%
\bibitem [{\citenamefont {B\'ar\'any}\ and\ \citenamefont
  {Karasev}(2012)}]{BK_2012}%
  \BibitemOpen
  \bibfield  {author} {\bibinfo {author} {\bibfnamefont {I.}~\bibnamefont
  {B\'ar\'any}}\ and\ \bibinfo {author} {\bibfnamefont {R.}~\bibnamefont
  {Karasev}},\ }\bibfield  {title} {\emph {\bibinfo {title} {Notes about the
  {C}arath\'eodory number},\ }}\href {\doibase 10.1007/s00454-012-9439-z}
  {\bibfield  {journal} {\bibinfo  {journal} {Discrete Comput. Geom.}\ }\textbf
  {\bibinfo {volume} {48}},\ \bibinfo {pages} {783} (\bibinfo {year}
  {2012})}\BibitemShut {NoStop}%
\bibitem [{\citenamefont {Beltrametti}\ \emph {et~al.}(2009)\citenamefont
  {Beltrametti}, \citenamefont {Carletti}, \citenamefont {Gallarati},\ and\
  \citenamefont {Monti~Bragadin}}]{beltrametti-book}%
  \BibitemOpen
  \bibfield  {author} {\bibinfo {author} {\bibfnamefont {M.~C.}\ \bibnamefont
  {Beltrametti}}, \bibinfo {author} {\bibfnamefont {E.}~\bibnamefont
  {Carletti}}, \bibinfo {author} {\bibfnamefont {D.}~\bibnamefont {Gallarati}},
  \ and\ \bibinfo {author} {\bibfnamefont {G.}~\bibnamefont {Monti~Bragadin}},\
  }\href@noop {} {\emph {\bibinfo {title} {Lectures on curves, surfaces and
  projective varieties}}}\ (\bibinfo  {publisher} {European Mathematical
  Society},\ \bibinfo {year} {2009})\BibitemShut {NoStop}%
\bibitem [{\citenamefont {Harris}(1992)}]{Harris_book}%
  \BibitemOpen
  \bibfield  {author} {\bibinfo {author} {\bibfnamefont {J.}~\bibnamefont
  {Harris}},\ }\href@noop {} {\emph {\bibinfo {title} {Algebraic geometry: A
  first course}}}\ (\bibinfo  {publisher} {Springer},\ \bibinfo {year}
  {1992})\BibitemShut {NoStop}%
\bibitem [{\citenamefont {Degen}(1994)}]{degen1994}%
  \BibitemOpen
  \bibfield  {author} {\bibinfo {author} {\bibfnamefont {W.~L.~F.}\
  \bibnamefont {Degen}},\ }\bibfield  {title} {\emph {\bibinfo {title} {The
  types of triangular {B}{\'e}zier surfaces},\ }}\href {\doibase
  10.5555/646872.709694} {\bibfield  {journal} {\bibinfo  {journal}
  {Proceedings of the 6th IMA Conference on the Mathematics of Surfaces}\ ,\
  \bibinfo {pages} {153}} (\bibinfo {year} {1994})}\BibitemShut {NoStop}%
\bibitem [{\citenamefont {Garrigue}(2022)}]{garrigue2021KS}%
  \BibitemOpen
  \bibfield  {author} {\bibinfo {author} {\bibfnamefont {L.}~\bibnamefont
  {Garrigue}},\ }\bibfield  {title} {\emph {\bibinfo {title} {Building
  {K}ohn-{S}ham potentials for ground and excited states},\ }}\href {\doibase
  10.1007/s00205-022-01804-1} {\bibfield  {journal} {\bibinfo  {journal} {Arch.
  Rational Mech. Anal.}\ }\textbf {\bibinfo {volume} {245}},\ \bibinfo {pages}
  {949} (\bibinfo {year} {2022})}\BibitemShut {NoStop}%
\bibitem [{\citenamefont {Apéry}(1987)}]{apery-book}%
  \BibitemOpen
  \bibfield  {author} {\bibinfo {author} {\bibfnamefont {F.}~\bibnamefont
  {Apéry}},\ }\href@noop {} {\emph {\bibinfo {title} {Models of the Real
  Projective Plane}}}\ (\bibinfo  {publisher} {Vieweg},\ \bibinfo {year}
  {1987})\BibitemShut {NoStop}%
\bibitem [{\citenamefont {Fortuna}\ \emph {et~al.}(2016)\citenamefont
  {Fortuna}, \citenamefont {Frigerio},\ and\ \citenamefont
  {Pardini}}]{fortuna2016book}%
  \BibitemOpen
  \bibfield  {author} {\bibinfo {author} {\bibfnamefont {E.}~\bibnamefont
  {Fortuna}}, \bibinfo {author} {\bibfnamefont {R.}~\bibnamefont {Frigerio}}, \
  and\ \bibinfo {author} {\bibfnamefont {R.}~\bibnamefont {Pardini}},\
  }\href@noop {} {\emph {\bibinfo {title} {Projective Geometry: Solved Problems
  and Theory Review}}},\ Vol.\ \bibinfo {volume} {104}\ (\bibinfo  {publisher}
  {Springer},\ \bibinfo {year} {2016})\BibitemShut {NoStop}%
\bibitem [{\citenamefont {Sederberg}\ and\ \citenamefont
  {Anderson}(1985)}]{sederberg1985steiner}%
  \BibitemOpen
  \bibfield  {author} {\bibinfo {author} {\bibfnamefont {T.}~\bibnamefont
  {Sederberg}}\ and\ \bibinfo {author} {\bibfnamefont {D.}~\bibnamefont
  {Anderson}},\ }\bibfield  {title} {\emph {\bibinfo {title} {Steiner surface
  patches},\ }}\href {\doibase 10.1109/MCG.1985.276391} {\bibfield  {journal}
  {\bibinfo  {journal} {IEEE Comput. Graph. Appl.}\ }\textbf {\bibinfo {volume}
  {5}},\ \bibinfo {pages} {23} (\bibinfo {year} {1985})}\BibitemShut {NoStop}%
\bibitem [{\citenamefont {Coffman}\ \emph {et~al.}(1996)\citenamefont
  {Coffman}, \citenamefont {Schwartz},\ and\ \citenamefont
  {Stanton}}]{Coffman-Roman-surf}%
  \BibitemOpen
  \bibfield  {author} {\bibinfo {author} {\bibfnamefont {A.}~\bibnamefont
  {Coffman}}, \bibinfo {author} {\bibfnamefont {A.}~\bibnamefont {Schwartz}}, \
  and\ \bibinfo {author} {\bibfnamefont {C.}~\bibnamefont {Stanton}},\
  }\bibfield  {title} {\emph {\bibinfo {title} {The algebra and geometry of
  {S}teiner and other quadratically parametrizable surfaces},\ }}\href
  {\doibase 10.1016/0167-8396(95)00026-7} {\bibfield  {journal} {\bibinfo
  {journal} {Comput. Aided Geom. Des.}\ }\textbf {\bibinfo {volume} {13}},\
  \bibinfo {pages} {257} (\bibinfo {year} {1996})}\BibitemShut {NoStop}%
\bibitem [{\citenamefont {Michel}(1926)}]{Michel-book}%
  \BibitemOpen
  \bibfield  {author} {\bibinfo {author} {\bibfnamefont {C.}~\bibnamefont
  {Michel}},\ }\href@noop {} {\emph {\bibinfo {title} {Compléments de
  géométrie moderne}}}\ (\bibinfo  {publisher} {Vuibert},\ \bibinfo {year}
  {1926})\BibitemShut {NoStop}%
\bibitem [{\citenamefont {Clebsch}(1867)}]{clebsch1867}%
  \BibitemOpen
  \bibfield  {author} {\bibinfo {author} {\bibfnamefont {A.}~\bibnamefont
  {Clebsch}},\ }\bibfield  {title} {\emph {\bibinfo {title} {Ueber die
  {S}teinersche {F}l{\"a}che.}\ }}\href@noop {} {\bibfield  {journal} {\bibinfo
   {journal} {Journal für die reine und angewandte Mathematik}\ }\textbf
  {\bibinfo {volume} {67}},\ \bibinfo {pages} {1} (\bibinfo {year}
  {1867})}\BibitemShut {NoStop}%
\bibitem [{\citenamefont {Cayley}(1873)}]{cayley1873steiner}%
  \BibitemOpen
  \bibfield  {author} {\bibinfo {author} {\bibfnamefont {C.}~\bibnamefont
  {Cayley}},\ }\bibfield  {title} {\emph {\bibinfo {title} {On {S}teiner's
  surface},\ }}\href {\doibase 10.1112/plms/s1-5.1.14} {\bibfield  {journal}
  {\bibinfo  {journal} {Proc. Lond. Math. Soc.}\ }\textbf {\bibinfo {volume}
  {1}},\ \bibinfo {pages} {14} (\bibinfo {year} {1873})}\BibitemShut {NoStop}%
\bibitem [{\citenamefont {Lacour}(1898)}]{lacour1898}%
  \BibitemOpen
  \bibfield  {author} {\bibinfo {author} {\bibfnamefont {E.}~\bibnamefont
  {Lacour}},\ }\bibfield  {title} {\emph {\bibinfo {title} {Sur la surface de
  {S}teiner},\ }}\href@noop {} {\bibfield  {journal} {\bibinfo  {journal}
  {Nouvelles annales de math{\'e}matiques: Journal des candidats aux {\'e}coles
  polytechnique et normale}\ }\textbf {\bibinfo {volume} {17}},\ \bibinfo
  {pages} {437} (\bibinfo {year} {1898})}\BibitemShut {NoStop}%
\bibitem [{\citenamefont {Hilbert}\ and\ \citenamefont
  {Cohn-Vossen}(2021)}]{hilbert-geometry-book}%
  \BibitemOpen
  \bibfield  {author} {\bibinfo {author} {\bibfnamefont {D.}~\bibnamefont
  {Hilbert}}\ and\ \bibinfo {author} {\bibfnamefont {S.}~\bibnamefont
  {Cohn-Vossen}},\ }\href@noop {} {\emph {\bibinfo {title} {Geometry and the
  Imagination}}},\ Vol.~\bibinfo {volume} {87}\ (\bibinfo  {publisher}
  {American Mathematical Society},\ \bibinfo {year} {2021})\BibitemShut
  {NoStop}%
\bibitem [{\citenamefont {Liu}\ \emph {et~al.}(2022)\citenamefont {Liu},
  \citenamefont {Pi}, \citenamefont {Zhou}, \citenamefont {Liu}, \citenamefont
  {Shen}, \citenamefont {Ye}, \citenamefont {Qin}, \citenamefont {Mi},
  \citenamefont {Chen}, \citenamefont {Zhao} \emph
  {et~al.}}]{liu2022physical-Roman-surf}%
  \BibitemOpen
  \bibfield  {author} {\bibinfo {author} {\bibfnamefont {G.}~\bibnamefont
  {Liu}}, \bibinfo {author} {\bibfnamefont {M.}~\bibnamefont {Pi}}, \bibinfo
  {author} {\bibfnamefont {L.}~\bibnamefont {Zhou}}, \bibinfo {author}
  {\bibfnamefont {Z.}~\bibnamefont {Liu}}, \bibinfo {author} {\bibfnamefont
  {X.}~\bibnamefont {Shen}}, \bibinfo {author} {\bibfnamefont {X.}~\bibnamefont
  {Ye}}, \bibinfo {author} {\bibfnamefont {S.}~\bibnamefont {Qin}}, \bibinfo
  {author} {\bibfnamefont {X.}~\bibnamefont {Mi}}, \bibinfo {author}
  {\bibfnamefont {X.}~\bibnamefont {Chen}}, \bibinfo {author} {\bibfnamefont
  {L.}~\bibnamefont {Zhao}},  \emph {et~al.},\ }\bibfield  {title} {\emph
  {\bibinfo {title} {Physical realization of topological {R}oman surface by
  spin-induced ferroelectric polarization in cubic lattice},\ }}\href {\doibase
  10.1038/s41467-022-29764-w} {\bibfield  {journal} {\bibinfo  {journal}
  {Nature Comm.}\ }\textbf {\bibinfo {volume} {13}},\ \bibinfo {pages} {2373}
  (\bibinfo {year} {2022})}\BibitemShut {NoStop}%
\bibitem [{\citenamefont {Barbu}\ and\ \citenamefont
  {Precupanu}(2012)}]{Barbu-Precupanu}%
  \BibitemOpen
  \bibfield  {author} {\bibinfo {author} {\bibfnamefont {V.}~\bibnamefont
  {Barbu}}\ and\ \bibinfo {author} {\bibfnamefont {T.}~\bibnamefont
  {Precupanu}},\ }\href@noop {} {\emph {\bibinfo {title} {Convexity and
  Optimization in Banach Spaces}}},\ \bibinfo {edition} {4th}\ ed.\ (\bibinfo
  {publisher} {Springer},\ \bibinfo {year} {2012})\BibitemShut {NoStop}%
\bibitem [{\citenamefont {Kvaal}\ \emph {et~al.}(2014)\citenamefont {Kvaal},
  \citenamefont {Ekström}, \citenamefont {Teale},\ and\ \citenamefont
  {Helgaker}}]{Kvaal2014}%
  \BibitemOpen
  \bibfield  {author} {\bibinfo {author} {\bibfnamefont {S.}~\bibnamefont
  {Kvaal}}, \bibinfo {author} {\bibfnamefont {U.}~\bibnamefont {Ekström}},
  \bibinfo {author} {\bibfnamefont {A.~M.}\ \bibnamefont {Teale}}, \ and\
  \bibinfo {author} {\bibfnamefont {T.}~\bibnamefont {Helgaker}},\ }\bibfield
  {title} {\emph {\bibinfo {title} {Differentiable but exact formulation of
  density-functional theory},\ }}\href {\doibase 10.1063/1.4867005} {\bibfield
  {journal} {\bibinfo  {journal} {J. Chem. Phys.}\ }\textbf {\bibinfo {volume}
  {140}},\ \bibinfo {pages} {18A518} (\bibinfo {year} {2014})}\BibitemShut
  {NoStop}%
\bibitem [{\citenamefont {Laestadius}\ \emph {et~al.}(2018)\citenamefont
  {Laestadius}, \citenamefont {Penz}, \citenamefont {Tellgren}, \citenamefont
  {Ruggenthaler}, \citenamefont {Kvaal},\ and\ \citenamefont
  {Helgaker}}]{laestadius2018generalized}%
  \BibitemOpen
  \bibfield  {author} {\bibinfo {author} {\bibfnamefont {A.}~\bibnamefont
  {Laestadius}}, \bibinfo {author} {\bibfnamefont {M.}~\bibnamefont {Penz}},
  \bibinfo {author} {\bibfnamefont {E.~I.}\ \bibnamefont {Tellgren}}, \bibinfo
  {author} {\bibfnamefont {M.}~\bibnamefont {Ruggenthaler}}, \bibinfo {author}
  {\bibfnamefont {S.}~\bibnamefont {Kvaal}}, \ and\ \bibinfo {author}
  {\bibfnamefont {T.}~\bibnamefont {Helgaker}},\ }\bibfield  {title} {\emph
  {\bibinfo {title} {Generalized {K}ohn--{S}ham iteration on {B}anach spaces},\
  }}\href {\doibase 10.1063/1.5037790} {\bibfield  {journal} {\bibinfo
  {journal} {J. Chem. Phys.}\ }\textbf {\bibinfo {volume} {149}},\ \bibinfo
  {pages} {164103} (\bibinfo {year} {2018})}\BibitemShut {NoStop}%
\bibitem [{\citenamefont {Levy}(1982)}]{Levy1982}%
  \BibitemOpen
  \bibfield  {author} {\bibinfo {author} {\bibfnamefont {M.}~\bibnamefont
  {Levy}},\ }\bibfield  {title} {\emph {\bibinfo {title} {Electron densities in
  search of {H}amiltonians},\ }}\href {\doibase 10.1103/PhysRevA.26.1200}
  {\bibfield  {journal} {\bibinfo  {journal} {Phys. Rev. A}\ }\textbf {\bibinfo
  {volume} {26}},\ \bibinfo {pages} {1200} (\bibinfo {year}
  {1982})}\BibitemShut {NoStop}%
\bibitem [{\citenamefont {Rellich}(1937)}]{rellich1937}%
  \BibitemOpen
  \bibfield  {author} {\bibinfo {author} {\bibfnamefont {F.}~\bibnamefont
  {Rellich}},\ }\bibfield  {title} {\emph {\bibinfo {title}
  {St{\"o}rungstheorie der {S}pektralzerlegung, {I}. {M}itteilung},\ }}\href
  {\doibase 10.1007/BF01571652} {\bibfield  {journal} {\bibinfo  {journal}
  {Mathematische Annalen}\ }\textbf {\bibinfo {volume} {113}},\ \bibinfo
  {pages} {600} (\bibinfo {year} {1937})}\BibitemShut {NoStop}%
\bibitem [{\citenamefont {Rellich}(1969)}]{rellich-book}%
  \BibitemOpen
  \bibfield  {author} {\bibinfo {author} {\bibfnamefont {F.}~\bibnamefont
  {Rellich}},\ }\href@noop {} {\emph {\bibinfo {title} {Perturbation theory of
  eigenvalue problems}}}\ (\bibinfo  {publisher} {Gordon and Breach Science
  Publishers},\ \bibinfo {year} {1969})\BibitemShut {NoStop}%
\bibitem [{\citenamefont {Kato}(1995)}]{Kato-book}%
  \BibitemOpen
  \bibfield  {author} {\bibinfo {author} {\bibfnamefont {T.}~\bibnamefont
  {Kato}},\ }\href@noop {} {\emph {\bibinfo {title} {Perturbation theory for
  linear operators}}}\ (\bibinfo  {publisher} {Springer},\ \bibinfo {year}
  {1995})\BibitemShut {NoStop}%
\bibitem [{\citenamefont {Penz}\ \emph {et~al.}(2019)\citenamefont {Penz},
  \citenamefont {Laestadius}, \citenamefont {Tellgren},\ and\ \citenamefont
  {Ruggenthaler}}]{penz2019guaranteed}%
  \BibitemOpen
  \bibfield  {author} {\bibinfo {author} {\bibfnamefont {M.}~\bibnamefont
  {Penz}}, \bibinfo {author} {\bibfnamefont {A.}~\bibnamefont {Laestadius}},
  \bibinfo {author} {\bibfnamefont {E.~I.}\ \bibnamefont {Tellgren}}, \ and\
  \bibinfo {author} {\bibfnamefont {M.}~\bibnamefont {Ruggenthaler}},\
  }\bibfield  {title} {\emph {\bibinfo {title} {Guaranteed convergence of a
  regularized {K}ohn--{S}ham iteration in finite dimensions},\ }}\href
  {\doibase 10.1103/physrevlett.123.037401} {\bibfield  {journal} {\bibinfo
  {journal} {Phys. Rev. Lett.}\ }\textbf {\bibinfo {volume} {123}},\ \bibinfo
  {pages} {037401} (\bibinfo {year} {2019})}\BibitemShut {NoStop}%
\bibitem [{\citenamefont {Penz}\ \emph {et~al.}(2020)\citenamefont {Penz},
  \citenamefont {Laestadius}, \citenamefont {Tellgren}, \citenamefont
  {Ruggenthaler},\ and\ \citenamefont {Lammert}}]{penz2020erratum}%
  \BibitemOpen
  \bibfield  {author} {\bibinfo {author} {\bibfnamefont {M.}~\bibnamefont
  {Penz}}, \bibinfo {author} {\bibfnamefont {A.}~\bibnamefont {Laestadius}},
  \bibinfo {author} {\bibfnamefont {E.~I.}\ \bibnamefont {Tellgren}}, \bibinfo
  {author} {\bibfnamefont {M.}~\bibnamefont {Ruggenthaler}}, \ and\ \bibinfo
  {author} {\bibfnamefont {P.~E.}\ \bibnamefont {Lammert}},\ }\bibfield
  {title} {\emph {\bibinfo {title} {Erratum: {G}uaranteed convergence of a
  regularized {K}ohn--{S}ham iteration in finite dimensions},\ }}\href
  {\doibase 10.1103/PhysRevLett.125.249902} {\bibfield  {journal} {\bibinfo
  {journal} {Phys. Rev. Lett.}\ }\textbf {\bibinfo {volume} {125}},\ \bibinfo
  {pages} {249902} (\bibinfo {year} {2020})}\BibitemShut {NoStop}%
\bibitem [{\citenamefont {Laestadius}\ \emph {et~al.}(2019)\citenamefont
  {Laestadius}, \citenamefont {Tellgren}, \citenamefont {Penz}, \citenamefont
  {Ruggenthaler}, \citenamefont {Kvaal},\ and\ \citenamefont
  {Helgaker}}]{laestadius2019CDFT}%
  \BibitemOpen
  \bibfield  {author} {\bibinfo {author} {\bibfnamefont {A.}~\bibnamefont
  {Laestadius}}, \bibinfo {author} {\bibfnamefont {E.~I.}\ \bibnamefont
  {Tellgren}}, \bibinfo {author} {\bibfnamefont {M.}~\bibnamefont {Penz}},
  \bibinfo {author} {\bibfnamefont {M.}~\bibnamefont {Ruggenthaler}}, \bibinfo
  {author} {\bibfnamefont {S.}~\bibnamefont {Kvaal}}, \ and\ \bibinfo {author}
  {\bibfnamefont {T.}~\bibnamefont {Helgaker}},\ }\bibfield  {title} {\emph
  {\bibinfo {title} {{K}ohn--{S}ham theory with paramagnetic currents:
  {C}ompatibility and functional differentiability},\ }}\href {\doibase
  10.1021/acs.jctc.9b00141} {\bibfield  {journal} {\bibinfo  {journal} {J.
  Chem. Theory Comput.}\ }\textbf {\bibinfo {volume} {15}},\ \bibinfo {pages}
  {4003} (\bibinfo {year} {2019})}\BibitemShut {NoStop}%
\bibitem [{\citenamefont {Laestadius}\ and\ \citenamefont
  {Tellgren}(2018)}]{LaestadiusTellgren2018}%
  \BibitemOpen
  \bibfield  {author} {\bibinfo {author} {\bibfnamefont {A.}~\bibnamefont
  {Laestadius}}\ and\ \bibinfo {author} {\bibfnamefont {E.~I.}\ \bibnamefont
  {Tellgren}},\ }\bibfield  {title} {\emph {\bibinfo {title}
  {Density--wave-function mapping in degenerate current-density-functional
  theory},\ }}\href {\doibase 10.1103/PhysRevA.97.022514} {\bibfield  {journal}
  {\bibinfo  {journal} {Phys. Rev. A}\ }\textbf {\bibinfo {volume} {97}},\
  \bibinfo {pages} {022514} (\bibinfo {year} {2018})}\BibitemShut {NoStop}%
\bibitem [{\citenamefont {Song}(2023)}]{Song2023}%
  \BibitemOpen
  \bibfield  {author} {\bibinfo {author} {\bibfnamefont {C.}~\bibnamefont
  {Song}},\ }\bibfield  {title} {\emph {\bibinfo {title} {Quantum geometry of
  expectation values},\ }}\href {\doibase 10.1103/physreva.107.062207}
  {\bibfield  {journal} {\bibinfo  {journal} {Phys. Rev. A}\ }\textbf {\bibinfo
  {volume} {107}} (\bibinfo {year} {2023}),\
  10.1103/physreva.107.062207}\BibitemShut {NoStop}%
\bibitem [{\citenamefont {Liebert}\ \emph {et~al.}(2023)\citenamefont
  {Liebert}, \citenamefont {Chaou},\ and\ \citenamefont
  {Schilling}}]{Liebert2023-1RDMFT}%
  \BibitemOpen
  \bibfield  {author} {\bibinfo {author} {\bibfnamefont {J.}~\bibnamefont
  {Liebert}}, \bibinfo {author} {\bibfnamefont {A.~Y.}\ \bibnamefont {Chaou}},
  \ and\ \bibinfo {author} {\bibfnamefont {C.}~\bibnamefont {Schilling}},\
  }\bibfield  {title} {\emph {\bibinfo {title} {Refining and relating
  fundamentals of functional theory},\ }}\href {\doibase 10.1063/5.0143657}
  {\bibfield  {journal} {\bibinfo  {journal} {J. Chem. Phys.}\ }\textbf
  {\bibinfo {volume} {158}} (\bibinfo {year} {2023}),\
  10.1063/5.0143657}\BibitemShut {NoStop}%
\end{thebibliography}
\end{document}